\documentclass[preprint,12pt]{elsarticle}

\usepackage{mathptmx}
\usepackage{amsmath}
\usepackage{amssymb}
\usepackage{amsthm}
\usepackage{mathrsfs}
\usepackage{graphicx}
\usepackage{bbm}
\usepackage{float}
\usepackage{subfig}
\usepackage{pifont}
\usepackage{color}

\theoremstyle{plain}
\newtheorem{prop}{Proposition}
\newtheorem{thm}{Theorem}

\theoremstyle{remark}
\newtheorem*{rem}{Remark}

\newcommand{\bX}{\mathbf{X}}
\newcommand{\bx}{\mathbf{x}}
\newcommand{\pb}{\mathbb{P}}
\newcommand{\lsq}{\left[}
\newcommand{\rsq}{\right]}

\newcommand{\pp}[1]{\left({#1}\right)}

\newcommand\cond{\, | \,}

\newcommand{\E}[1]{\mathbb{E}\!\left[#1\right]}
\newcommand{\Var}[1]{\mathrm{Var}\left({#1}\right)}
\newcommand{\rd}{{\mathrm d}}

\newcommand{\abs}[1]{\left\lvert#1\right\rvert}
\newcommand{\meanT}{m_T}
\newcommand{\meanL}{m_L}

\journal{Journal of Theoretical Biology}

\begin{document}

\begin{frontmatter}



\title{Exact and approximate moment closures for non-Markovian network epidemics}


\author[WMI]{Lorenzo Pellis\corref{cor}}
\author[Manc,WMI]{Thomas House}
\ead{thomas.house@manchester.ac.uk}
\author[WMI,SLS]{Matt J. Keeling}
\ead{M.J.Keeling@warwick.ac.uk}

\address[WMI]{Warwick Mathematics Institute, University of Warwick, Coventry, CV4 7AL, UK.}
\address[SLS]{School of Life Sciences, University of Warwick, Coventry, CV4 7AL, UK.}
\address[Manc]{School of Mathematics, University of Manchester, Manchester, M13 9PL, UK.}

\cortext[cor]{Corresponding author: \texttt{L.Pellis@warwick.ac.uk}}

\begin{abstract}
Moment-closure techniques are commonly used to generate low-dimensional
deterministic models to approximate the average dynamics of stochastic systems
on networks. The quality of such closures is usually difficult to asses and
furthermore the relationship between model assumptions and closure accuracy are
often difficult, if not impossible, to quantify. Here we carefully examine some
commonly used moment closures, in particular a new one based on the concept of
maximum entropy, for approximating the spread of epidemics on networks by
reconstructing the probability distributions over triplets based on those over
pairs. We consider various models (SI, SIR, SEIR and Reed-Frost-type) under
Markovian and non-Markovian assumption characterising the latent and infectious
periods. We initially study with care two special networks, namely the open
triplet and closed triangle, for which we can obtain analytical results. We
then explore numerically the exactness of moment closures for a wide range of
larger motifs, thus gaining understanding of the factors that introduce errors
in the approximations, in particular the presence of a random duration of the
infectious period and the presence of overlapping triangles in a network. We
also derive a simpler and more intuitive proof than previously available
concerning the known result that pair-based moment closure is exact for the
Markovian SIR model on tree-like networks under pure initial conditions. We
also extend such a result to all infectious models, Markovian and
non-Markovian, in which susceptibles escape infection independently from each
infected neighbour and for which infectives cannot regain susceptible status,
provided the network is tree-like and initial conditions are pure. This works
represent a valuable step in enriching intuition and deepening understanding of
the assumptions behind moment closure approximations and for putting them on a
more rigorous mathematical footing.
\end{abstract}


\end{frontmatter}


\section{Introduction}
\label{sec:Intro}

Networks are becoming a ubiquitous tool for modelling the interactions
between systems of multiple components with complex interactions
between them~\cite{Boccaletti:2006,Durrett:2007,Newman:2010}. This is
particularly true for epidemic models, where empirical advances in measurement
of relevant interactions are acting as a particular spur to theoretical
developments~\cite{Danon:2011,Read:2012}.

One particular challenge for complex network modelling consists in the high
dimensionality of the dynamical systems. If a network has $N$ nodes, each of
which can be in one of $m$ states, then the dimensionality of a stochastic
process for the evolution of those states will be $\mathcal{O}(m^N)$ in the
absence of a large discrete symmetry group for the network or a specific combination of
dynamical and network models that allows for analytical results to be obtained
(such as SIR dynamics on an Erd\"{o}s-R\'{e}nyi
graph~\cite{Neal:2003}).

For general dynamics and topologies, network moment closure techniques provide a
commonly used method of gaining significant dimensional reduction, at the price
of losing an exact description of the system dynamics. These closures are based on the
idea of approximating the dynamics of small subgraphs in the network (e.g.\
adjacent pairs of nodes) by forcing their time derivatives to depend only on
the state of subgraphs of the same or lower dimension instead of on the state of larger subgraphs (e.g.\
triplets), thus deriving a closed
system of equations.  Every closure therefore implicitly makes assumptions about the
probability distribution over states of certain parts of the system, in terms
of probability distributions over states of smaller parts of the system,
although these are often not stated explicitly. If we exclude the case of
constructing joint probabilities over pairs of adjacent nodes as the products
of marginals over the single nodes (which leads to the so-called
\emph{mean-field approximation}; see e.g.~\citep{Sharkey:2008}), the next most common moment-closure
approximation involves describing the probability of triplets of adjacent nodes
being in any possible state based on the knowledge of the probabilities over
pairs. This has lead to many so-called \emph{pair approximation} models, widely
used both for theoretical~\cite{Keeling:1999,Trapman:2007} and
practical~\cite{Ferguson:2001,EamKee2002,House:2010} purposes.

Unfortunately, the overall quality of a moment closure approximation is often
difficult to asses, thus severely limiting the generalisability of results
based on such approaches. Even when proved accurate in certain cases, it is not 
clear whether such accuracy is preserved in other slightly different contexts. 
Furthermore, moment-closure approximations, often proposed in an \textit{ad hoc} 
fashion on the
basis of heuristic arguments, unavoidably impose assumptions on the
interactions between system components. The fact that such assumptions are
often obscure, and their interaction and impact on the other modelling
assumptions are far from trivial to unravel, compromises the neatness of the
approach and makes them somewhat less appealing from a theoretical point of
view. 

Consequently, there is significant interest within the research community in
deepening our understanding of which moment closure approximations are more
accurate than others, when they fail to reproduce the system dynamics exactly
and why~\cite{Sharkey:2008,Karrer:2010,Sharkey:2011,Taylor:2011,Sharkey:2012}. As a first step in this direction, a recent trend has involved applying moment closures to each specific subgraph of interest in order to approximate its dynamics. This framework, on a large network, results in a fairly large number of equations. \citet{Sharkey:2008} refers to this modelling approach as \emph{individual-based} or \emph{pair-based}, depending on whether the aim is to describe only the dynamics of each single node or of pairs of nodes. We propose to collectively refer to it as \emph{local} network moment closure. On the other hand, the original and most commonly used type of moment closure consists in counting the \emph{number} of subgraphs of interest in any possible configuration at any one time and is where the most significant dimensionality reduction is gained \citep{Rand:1999,Keeling:1999,House:2011}. \citet{Sharkey:2008} refers to this other approach as \emph{mean-field approximation} or \emph{pair approximation}, depending on whether the interest is on single node or pairs of nodes. We suggest referring collectively to this framework as \emph{global} (or \emph{population-level}) network moment closure. 

Scaling up from local to global moment closure introduces a further round of approximation, on top of the one already present at the local level. \citet{Sharkey:2008} points out how this second level of approximation depends on an averaging or ``mean-field'' assumption of homogeneity, the accuracy of which depends primarily on the heterogeneity in the network structure, more than on the dynamical errors built in at the local level. Therefore, as a first step in gaining better understanding of the quality of moment closure approximations in general, here we focus only on local moment closures and the dynamical local errors they generate.
   

In the specific context of local moment closure approximations for Susceptible-Infected-Recovered (SIR)
epidemic models on a network, recent work by~\citet{Sharkey:2012}
has shown that, provided the network has no short loops and initial conditions are pure (i.e.~the system start in a specific state with probability 1), the standard loopless
pair-based local moment closure (see~\cite{Sharkey:2008}) provides an exact description of the dynamics of single nodes and pairs of nodes, from which, for example, the expected epidemic course can be obtained exactly. When the
network does have short loops, in particular triangles, other closure
techniques have been proposed. The most common of these is due to Kirkwood~(\citep{Kirkwood:1935}; see also \citep{Sharkey:2008,Sharkey:2011}), which can often be quite accurate in practice,
but lacks solid theoretical justification.

Recently, work has been done to provide more explicit derivations of
novel moment closures in the presence of closed loops. This has included
arguments about appropriate early asymptotic behaviour~\cite{House:2010pcb},
non-independent Bernoulli trials~\cite{Taylor:2011} and maximum
entropy~\cite{Rogers:2011}. It turns out that these are equivalent at `first
order', but the maximum entropy (ME) approach is more readily generalisable and
can be used to derive a large variety of moment closures.

In this paper we carefully investigate the behaviour of various local moment closure techniques for reconstructing the
behaviour of triplets in terms of pairs on networks of increasing size and complexity and try to clarify
when and why they lack exactness and for which modelling assumptions.

In Section \ref{sec:GeneralFramework} we introduce the notation and describe
the basic model assumptions for all models considered in the paper. In Section
\ref{sec:MomentClosureApproximations} we define the moment closure
approximations studied and we propose a different and possibly more intuitive
interpretation of the ME approximation. In Sections \ref{sec:Markovian}, \ref{sec:Const} and \ref{sec:GammaDistributedInfectiousPeriod} we focus on the SIR model on the simplest possible network topologies, namely an open triplet and a closed triangle, and show how the behaviour of the moment closures change when changing the assumptions about the distribution of the infectious period. In particular, in Section \ref{sec:Const}, we show that for such simple structures, all the
moment closure approximations that we consider here are exact when the infectious period has a
constant duration.  When the duration is random, as is the case for the
Markovian model, the closure is in general only approximate, although the most
important quantities for the open triplet are still captured exactly. In Section
\ref{sec:GammaDistributedInfectiousPeriod} we explore the convergence for all
approximations to the exact results as the variance in the duration of the
infectious period tends to 0, using a family of non-Markovian epidemic models
with Erlang-distributed durations of the infectious period. Furthermore, we
highlight the overall superior accuracy of the moment closure technique based
on ME, but shed light on its context-specific limitations in comparison with
the other closures. In Section \ref{sec:OtherMotifs} we explore how results extend to slightly larger structures and build up the intuition about when the closure considered here are exact on larger networks. Such intuition is then discussed in Section \ref{sec:LargeNetworks}, where we conjecture how the errors introduced by moment closure behave on larger networks and also prove that the standard pair-based local approximation is exact on tree-like networks with pure initial conditions for all models considered here, thus extending and simplifying result already known for the SIR Markovian epidemic model.

\section{General framework}
\label{sec:GeneralFramework} 

\subsection{Labelled network}
\label{sec:LabelledNetwork}

We consider an undirected static network $\mathcal{G} = (\mathcal{N},
\mathcal{L})$, which has a size-$N$ set of nodes $\mathcal{N}$ and a set of links
$\mathcal{L}$. Nodes are denoted by $i,j,\ldots \in \mathcal{N}$ and $\{i,j\}
\in \mathcal{L}$ if and only if $i$ and $j$ are connected to each other (and we use the
convention $\{i,i\}\notin\mathcal{L}$). 

At any time $t$, each node $i$ is labelled by a state $X_i(t) \in \Omega$, where $\Omega$ is a set of states that depends on the epidemic model run on the network (see below; for example $\Omega = \{S,I,R\}$). We will assume throughout that the network structure is not affected by the states of its nodes.

Let $\bX(t)=(X_1(t),X_2(t),\dots,X_N(t))$ be a vector describing the random state of the system at time $t$. We denote by $\bx=(x_1,x_2,\dots,x_N)$ ($x_i\in\Omega, i=1,2,\dots,N$) a
specific system state, and let $\bx^0=(x_1^0, x_2^0, \dots, x_N^0)$ denote the
initial state. Then the state of the system at each time $t\geq 0$ is
described by the probability distribution
\begin{equation} \label{systemstate}
\pb^{\bx_0}\pp{\bx;t} = \pb\pp{\bX(t)=\bx \cond \bX(0)=\bx^0} \; .
\end{equation}
Note that, in general, a process is not fully specified by its marginal distributions over time\footnote{For example, the three-state Markov chains with generator matrices
\[M_1 = \begin{pmatrix}
-1 & 1/2 & 1/2 \\
0 & 0 & 0 \\
0 & 0 & 0
\end{pmatrix}\qquad\text{and}\qquad
M_2 = \begin{pmatrix}
-1 & 1/2 & 1/2 \\
0 & -1 & 1 \\
0 & 1 & -1
\end{pmatrix}\] have the same marginal distributions at any time, if they both start in the first state. However, their behaviour is different.}. However, for our purposes Equation \eqref{systemstate} is sufficient.

If we are interested in the state of a subsystem, we first consider the set $\mathcal{V}\subset\mathcal{N}$ of all indices of the nodes we are interested in. Upon choosing a reference ordering, thus replacing the set $\mathcal{V}$ with a vector $V$, we then consider the vectors $\bX_V$ and $\bx_V$ which contain only the
elements of $\bX$ and $\bx$ with indices in $V$. Applying the subscript $V$ can
be thought as a projection on the subspace identified by $V$ of the
$N$-dimensional space $\Omega^N$. By definition,
\begin{equation} \label{vstate}
\pb^{\bx_0}_V\pp{\bx_V;t} = \pb\pp{\bX_V(t)=\bx_V\cond \bX(0)=\bx^0}
\end{equation}
is obtained by summing \eqref{systemstate} over all indices not appearing in
$V$. Note that the initial conditions should remain specified on the full
graph.

\subsection{Epidemic models}
\label{sec:EpiModels}
We are interested in the spread of an epidemic on the static network described above. We consider different epidemic models, namely an SI, an SIR, an SEIR and a Reed-Frost model. In all models, the epidemic spreads by infective ($I$) nodes transmitting the infection to susceptible ($S$) neighbours.

\subsubsection{SI model}
\label{sec:EpiModelsSI}

In the SI model, $\Omega = \{S,I\}$. Upon infection, node $i$
makes infectious contacts to each one of its neighbours at the 
points of a homogeneous Poisson process with rate $\tau > 0$. A contacted
node, if susceptible, becomes infectious. Therefore, the epidemic 
results in the infection of all nodes in the connected 
components containing at least one initial infective node, and 
every infective ultimately infects all of its neighbours.

\subsubsection{SIR model}
\label{sec:EpiModelsSIR}

In the SIR model, $\Omega = \{S,I,R\}$. Upon infection, node $i$ is assumed to experience an infectious
period of random (non-negative) duration $T_i$, and during its infectious
period, it makes infectious contacts with each one of its neighbours at the
points of a homogeneous Poisson process with rate $\tau\geq 0$. A contacted
node, if susceptible, becomes infectious, and at the end of the infectious
period, the node recovers ($R$) and becomes permanently immune to the infection. The
infectious periods and Poisson processes associated with different infectious nodes
are assumed to be mutually independent; similarly, the Poisson processes from
the same infectious node towards different neighbours are mutually independent, conditionally on its
infectious period. We assume that for all $i$, the
random variables $T_i$ are independent and identically distributed (iid)
according to a random variable $T$ with mean by $\meanT = \mathbb{E}\left[ T \right]$. Without loss of generality, we assume in
all numerical examples that $\meanT = 1$ and, unless stated otherwise, that $\tau = 1$.

\subsubsection{SEIR model}
\label{sec:EpiModelsSEIR}

In the SEIR model,  $\Omega = \{S,E,I,R\}$. This model is similar to the SIR one, with the additional presence of a latent period ($E$) following the infection of each node $i$, of duration $L_i$. During the latent period, a node cannot transmit the infection and will eventually progress to the infectious stage. The latent periods of different nodes are iid according to a random variable $L$ with mean $\meanL = \E{L}$, and are assumed to be independent of all infectious periods and Poisson processes describing infectious contacts, irrespective of whether they refer to the same node or different nodes.

\subsubsection{Reed-Frost-type models}
\label{sec:EpiModelsRF}

In the standard Reed-Frost (RF) model, $\Omega = \{S,E,R\}$. Upon infection, node $i$
experiences a latent period, at the end of which it spreads all its
infectivity at a single point in time and then recovers permanently. We
consider extensions of the standard Reed-Frost model to both a random
duration of the latent period and random probabilities of transmission. More
specifically, we assume that node $i$'s latent period has duration $L_i$ (iid
for different nodes according to $L$, with mean $\meanL$) and we denote by
$P_i$ the random probability with which node $i$ can infect each one of its
neighbours. All the $P_i$s are iid according to a random variable $P$, with
mean $p=\E{P}$, and are independent of latent periods, whether referring to the
same or to different nodes. Note that the infections of different neighbours by
node $i$ are not independent events, but are independent conditionally on the
value of $P_i$. In the literature, the term \emph{Reed-Frost model} refers only
to the case where the latent period is of fixed duration and $P$ is non-random
(i.e.~$L\equiv \meanL$ and $P\equiv p$), and the term \emph{randomised
Reed-Frost model} refers to a constant latent period $L\equiv \meanL$ and a
random probability of transmission $P$. Here, therefore, we
refer to all possible combinations of random $L$ and $P$ as
\emph{Reed-Frost-type models}. Note that
Reed-Frost-type models can be viewed as
limiting cases of SEIR models
where $\meanT\to 0$ while $\tau\to\infty$, such that the mean probability of
transmission $p = \tau \meanT$ is kept constant, for suitably chosen
distributions for the sojourn times in states E and I.

\subsection{Moment closures}
\label{sec:MomClos}

A moment closure, $\alpha$ say, is a rule for the generation of a probability
distribution for a set $\mathcal{V}$ of nodes in $\mathcal{N}$ from the probability distributions over subsets of $\mathcal{V}$. To avoid trivial cases, we implicitly assume that the subgraph identified by $\mathcal{V}$ (i.e.~consisting of the nodes in $\mathcal{V}$ and all and only the edges between nodes in $\mathcal{V}$) is connected. Again, we find it easier to specify an order for the nodes in $\mathcal{V}$, thus effectively listing them in a vector $V$. Whether or not a specific moment closure is exact for
some $t \geq 0$, i.e. whether
\begin{equation}\label{GenericAppWorks}
\pb_{V,\alpha}^{\bx^0} (\bx_V;t)=\pb_V^{\bx^0} (\bx_V;t) \; ,
\end{equation}
in general depends on the particular choice of $\bx_V$ and $\bx^0$. So, in what
follows, with the notation $\bx^0 \to \bx_V$, we generically refer to the
investigation of the evolution of the system from the initial state $\bx^0$ to
the state $\bx_V$. 

If clear from context, the initial condition $\bx^0$ will
often be removed. We will also drop the explicit dependence on $t$, implicitly
assuming that equalities are meant to hold for all $t\geq 0$ and inequalities
are meant to indicate that the corresponding equality fails for at least one
value of $t\geq 0$.

The main focus of this paper is on pair-based approximations to epidemic
dynamics on graphs of size 3, i.e.~on local
moment closures where the probability of the vector $V$ of three nodes being in
any possible configuration is reconstructed from the probability of single
nodes and pairs of nodes in $V$ being in any possible configuration. Various
common possible choices are carefully described in Section
\ref{sec:MomentClosureApproximations}.


Given that in this context $V$ will
have $\leq 3$ indices and we are often interested in indicating them explicitly, we will further simplify the notation by writing, 
for example for $V = \pp{1,2,3}$, 
\begin{align*}
\pb_{123}(ABC) &\quad \text{instead of}\quad\pb_{\pp{1,2,3}}(ABC)  \quad \text{and}\\
\pb_{123,\alpha}(ABC) & \quad \text{instead of}\quad\pb_{\pp{1,2,3},\alpha}(ABC) \qquad (A,B,C \in \Omega) .
\end{align*}
If obvious from the context which vector $V$ of three nodes is under consideration, we will often simply denote the probability over $V$ as $\pb(ABC)$.


\subsection{Explicit topologies}
\label{sec:ExplicitTopologies}

In order to develop understanding of the impact of the assumptions behind moment closure approximations, we consider numerous simple topologies. These small networks are presented in Figure \ref{fig:graphs}. However, we first begin with a careful study of the behaviour of the considered pair approximations in the context of the SIR model on the simplest possible graphs, namely the open triplet and closed triangle. These are the focus of Sections \ref{sec:Markovian}, \ref{sec:Const} and \ref{sec:GammaDistributedInfectiousPeriod}, on which then the other  sections are built.

In both the open triplet and the closed triangle we have $\mathcal{N} = \{1,2,3\}$, but different sets of links:
\begin{equation}
\mathcal{L}_{\mathrm{open}} = \{ \{1,2\} , \{2,3\} \} \; , \qquad
\mathcal{L}_{\mathrm{closed}} = \{ \{1,2\} , \{2,3\} , \{3,1\} \} \; .
\end{equation}
Figure~\ref{fig:sirmodel} shows the states and transitions of these explicit
models -- even for these small networks, there is a lot of dynamical structure
for any moment closure to capture.

\section{Moment closure approximations}\label{sec:MomentClosureApproximations}

In this section we illustrate all the moment closures we consider in this analysis. Note that any vector $V$ of three distinct nodes of a connected graph either forms an open triplet or a closed triangle. Also, without loss of generality, we assume throughout this section that $V=( 1, 2, 3 )$.

\subsection{Unclustered closure}
\label{sec:MCopen}

In the literature, the most common moment closure approximation of triplets in
terms of pairs is obtained following the na\"{i}ve idea of multiplying the
probability of every pair of nodes linked by an edge, and then dividing by the
probability of nodes common to pairs of edges~\cite{Sharkey:2008,Sharkey:2011}. On
an open triplet, 
assuming that $i=2$ is the central
node, this approximation, hereafter denoted by $o$, is defined as
\begin{equation}\label{MEopen}
\pb_{o} (ABC)= \pb_{123,o} (ABC)= \frac{\pb_{12}(AB)\pb_{23}(BC)}{\pb_2(B)}\; .
\end{equation}

\subsection{Kirkwood closure}
\label{sec:MCkappa}

On a closed triangle, the same approach leads to the following approximation,
popularised in epidemic modelling by \cite{Keeling:1999} and sometimes attributed
to Kirkwood~\cite{Kirkwood:1935} (see also \citep{Sharkey:2008,Sharkey:2011}), which we denote by $\kappa$:
\begin{equation}\label{Kapp} 
\pb_{\kappa}(ABC) = \frac{\pb_{12}(AB) \pb_{23}(BC) \pb_{13}(AC)}{\pb_1(A) \pb_2(B)\pb_3(C)}
\; .
\end{equation}
Kirkwood's approximation has the natural property of being symmetric in $A,B$
and $C$ but it is not always a proper distribution over system states (i.e.~sometimes $\sum_{a,b,c}\pb_\kappa(abc) \neq 1$) and it
does not always agree with the marginals it is constructed from (i.e.~$\sum_c \pb_\kappa(ABc)$ is in general different from $\pb_{12}(AB)$; see \cite{Rogers:2011,Sharkey:2011}). 

\subsection{Maximum entropy}
\label{sec:MCmu}

In order to overcome these limitations, Rogers \cite{Rogers:2011} recently
suggested constructing an approximation based on the principle of Maximum
Entropy (ME), which we here denote by $\mu$. In our context, this means that the quantity
\begin{equation}\label{EntDef}
E := - \sum_{a,b,c} \pb_\mu(abc) \ln ( \pb_\mu(abc) ) \; ,
\end{equation}
which is the information entropy of the distribution $\pb_\mu$, is maximised subject
to the constraints imposed by the marginals $\{\pb_i(A), \pb_{ij}(AB)\}$, i.e.~that $\pb_1(A) = \sum_{b,c}\pb_\mu(Abc)$ and $\pb_{12}(AB) = \sum_c{\pb_\mu(ABc)}$ (and similarly for all other nodes or pairs). For
the open triplet, the closure~\eqref{MEopen} is the ME distribution. For
the closed triangle there is no closed-form solution, although following
Rogers~\cite[Eq.\ 4]{Rogers:2011}, we know that a set of functions
$\{q_{ij}\}$, $\{i,j\}\in\mathcal{L}_{\text{closed}}$, exists such that the ME distribution can be written in product form 
\begin{equation}\label{EntQ}
\pb_{\mu}(ABC) = q_{12}(AB) q_{23}(BC) q_{31}(CA) \; .
\end{equation}
These functions are not, however, straightforwardly related to the marginal
probabilities and so an alternative approach is preferable for explicit
calculations.

\subsection{Iterative scaling}
\label{sec:MCiter}

Rogers \cite{Rogers:2011} provides an iterative scheme to calculate the ME
distribution: start with the uniform distribution $\pb^{(0)}(\bx)=1/\abs{\Omega}^3$ over all
possible system states $\bx \in \Omega^3$ and cycle through all the three
pairs (in any order; we choose the order $V_1 = \pp{1,2}, V_2 = \pp{2,3}, V_3 = \pp{1,3}$) to
obtain, for $n = 0,1,2,\dots$:
\begin{equation}\begin{aligned} \pb^{(n),V_1}(ABC) & =
\pb_{12}(AB)\frac{\pb^{(n)}(ABC)}{\sum_{c\in\Omega}\pb^{(n)}(ABc)} \; , \\
\pb^{(n),V_2}(ABC) & =
\pb_{23}(BC)\frac{\pb^{(n),V_1}(ABC)}{\sum_{a\in\Omega}\pb^{(n),V_1}(aBC)} \; ,
\\ \pb^{(n+1)}(ABC) = \pb^{(n),V_3}(ABC) & =
\pb_{13}(AC)\frac{\pb^{(n),V_2}(ABC)}{\sum_{b\in\Omega}\pb^{(n),V_2}(AbC)} \; . \label{itscale}
\end{aligned}\end{equation}
\citet{Rogers:2011} cites results from \citet{Csiszar:2004} to argue that the sequence
\[
\pb^{(0)}(ABC),\dots,\pb^{(n),V_1}(ABC),\pb^{(n),V_2}(ABC), \pb^{(n),V_3}(ABC),\pb^{(n+1),V_1}(ABC),\dots
\] 
converges as $n\to\infty$ and that the limiting distribution is the ME distribution $\pb_\mu(ABC)$, which is known to be unique (provided the marginals are consistent).

If a closed-form approximation is needed, \citet{Rogers:2011} suggests
using what is obtained from the algorithm after the first (triple) step.
Denoting this 1-step ME approximation with $\rho$, we have:
\begin{equation}\label{Rhoapp}
\pb_\rho(ABC) = \pb^{(1)}(ABC) =
\frac{\pb_{12}(AB) \pb_{23}(BC) \pb_{13}(AC)}{\pb_2(B)
\sum_{b}{\frac{\pb_{12}(Ab)\pb_{23}(bC)}{\pb_2(b)}}}\; .
\end{equation}
This closure has also been derived independently from different arguments in 
\cite{House:2010pcb} and \cite{Taylor:2011}. On the open triplet, the 1-step ME 
approximation \eqref{Rhoapp} leads once more to \eqref{MEopen}.  On the closed 
triangle, it overcomes the key limitations of Kirkwood's, i.e.~it is a proper 
distribution over system states and has the correct marginals. However it depends on the arbitrary choices of the starting distribution and the order in which to cycle through the pairs. On the contrary, $\pb_\mu(ABC)$ does not depend on either of these choices \citep{Csiszar:2004}. Note that other distributions are possible, for which some or all the requirements above are satisfied (for example, the algebraic mean of the six possible forms of 1-step ME, one per permutation order through the pairs, would be independent of the cycling order). However, among all distributions, ME is the only one that introduces no additional (and hence unjustifiable) information apart from the desired constraints, and is therefore the most theoretically appealing one.

We finally suggest a different and, to our knowledge, novel formulation that may provide a different point of view of the assumptions underlying maximum entropy. A simple and systematic means of generating triangles (that can be readily extended to other networks) is to derive iteratively a set of functions $\{\hat{q}_{ij}\}$ through the following procedure. Denoting by $q_{ij}^{(n)}(AB)$ the approximation of $\hat{q}_{ij}(AB)$ obtained at the $n^{\text{th}}$ iteration, we start with $q_{ij}^{(0)}(AB) = \pb_{ij}(AB)$. Then for the iterative step, we define
\begin{equation*}\label{pairsinabag}
\pi^{(n)}(ABC)=\frac{q_{12}^{(n)}(AB) q_{23}^{(n)}(BC) q_{13}^{(n)}(AC)}{\sum_{a,b,c} q_{12}^{(n)}(ab) q_{23}^{(n)}(bc) q_{13}^{(n)}(ac)}\; .
\end{equation*}
If $\sum_c \pi^{(n)}(ABc)>q_{12}^{(n)}(AB)$, then $q_{12}^{(n)}(AB)$ is updated to a new lower value $q_{12}^{(n+1)}(AB)$ while, if $\sum_c \pi^{(n)}(ABc)<q_{12}^{(n)}(AB)$, then one has to set $q_{12}^{(n+1)}(AB)>q_{12}^{(n)}(AB)$, where the change in value is determined by questions of numerical efficiency. Assuming convergence, we define
\begin{equation}\label{EntQ2}
\pb_{\hat{\mu}}(ABC) = \frac{\hat{q}_{12}(AB) \hat{q}_{23}(BC) \hat{q}_{31}(CA)}{ \sum_{a,b,c} \hat{q}_{12}(ab) \hat{q}_{23}(bc) \hat{q}_{31}(ca) }\; .
\end{equation}
These $\{\hat{q}_{ij}\}$ differ from the $\{q_{ij}\}$ in \eqref{EntQ} by virtue of being probability distributions, although clearly \eqref{EntQ2} is identical to (\ref{EntQ}) if the denominator is absorbed into the individual probabilities, and the existence and uniqueness of $\hat{q}$ implicitly assumed follow from the results of \cite{Rogers:2011}.

While such an argument offers a different route to the same result, we found that the iterative scaling approach outlined in \eqref{itscale} above is computationally more efficient (in addition to having been proved to converge). An implementation of the iterative scaling approach in \textsc{Matlab} is provided as Electronic Supplementary Material.

\section{Markovian SIR model on three nodes}
\label{sec:Markovian}

Here, as well as in Sections \ref{sec:Const}, and \ref{sec:GammaDistributedInfectiousPeriod}, we specifically focus on the performance of pair approximation on the open triplet and the closed triangle, where nodes are labelled $i=1,2,3$ and $i=2$ is the middle node in the triplet.

The majority of epidemic models appearing in the literature that make use of moment
closure assume that $T \sim \text{Exp}(\gamma)$, for some constant $\gamma>0$, and are therefore
fully Markovian. The main reason is mathematical convenience and a set
of ordinary differential equations is then derived to describe the probability of triplets being in each configuration of interest (local moment closure) or to describe the average behaviour
of the original stochastic model in the limit of a infinite population (global moment closure).

Therefore, in Figure \ref{fig:MarkExamples} we explore how all approximations
above ($\kappa,\rho,$ and $\mu$) compare to the exact probability distributions $\pb^{\bx^0}(\bx;t)$ on
the open and closed triangle at time $t=1$, for some natural choices of $\bx^0$
and $\bx$, when the Markovian model is used and $\tau=1$ and $\meanT=1$.  This
Figure shows that, on the closed triangle, no approximation is exact for any
state. The case of the open triplet is more subtle: for example,
$\pb^{\bx^0}(\bx;t)$ is different from $\pb_o^{\bx^0}(\bx;t)$ as defined in
\eqref{MEopen} when $\bx^0 = (SIS)$ and $\bx = (SRS)$. The reason is that the
random duration of the infectious period imposes correlations between the two
susceptibles even if there is no direct link between them. In fact, denote by
$Q$ the probability that a susceptible escapes infection when $t\to\infty$.
Then, $\lim_{t\to\infty} \pb^{\bx^0} (\bx;t) = \E{Q^2}$, which is in general different from $\lim_{t\to\infty} \pb_o^{\bx^0} (\bx;t) = \E{Q}^2$, except when $Q$ is non-random (e.g.~constant duration of infection).
Intuitively, if individual 2 has
recovered without infecting individual 1, then it is more likely that the
infectious period was shorter than expected, which in turn increases the
probability that also individual 3 escaped infection. Therefore, the joint
probability that both have escaped infection ($\pb^{(SIS)}(SRS)$) is higher
than that obtained through \eqref{MEopen}, where the two are assumed to escape
infection independently of each other. For the same reason, it is possible to
verify that the ME approximation also underestimates $\pb^{(SIS)}(RRR)$ and
overestimates $\pb^{(SIS)}(RRS)$ and $\pb^{(SIS)}(SRR)$. This insight suggests that the qualitative features that can be drawn from Figure \ref{fig:MarkExamples} are not exclusive to the Markovian model, but extend to all models with a random duration of infectious periods.

Analogously to the fact that $\pb_o^{(SIS)}(SRS) \neq \pb^{(SIS)}(SRS)$, in the presence of a
random infectious period we also have that (see Figure \ref{fig:MarkExamples})
\begin{equation}\label{MarkMEfails}
\begin{array}{rcl}
\pb_o^{(ISS)}(IIS) & \neq & \pb^{(ISS)}(IIS) \; , \\ 
\pb_o^{(ISI)}(III) & \neq & \pb^{(ISI)}(III) \; , \\ 
\end{array} 
\end{equation}
and therefore approximation $o$ on the open triplet fails to be exact for all
$t>0$ in these cases and all those where the system can evolve to from $(IIS)$
and $(III)$ (e.g. $(RIS)$ or $(RRI)$). However, even in the Markovian case (see Figure
\ref{fig:MarkExamples}), we have a set of equations that hold true:
\begin{equation}\label{MarkMEworksS}
\begin{array}{rcl}
\pb_o^{(ISS)}(ISS) & = & \pb^{(ISS)}(ISS) \; , \\ 
\pb_o^{(ISS)}(ISI) & = & \pb^{(ISS)}(ISI) \; , \\ 
\pb_o^{(ISS)}(ISR) & = & \pb^{(ISS)}(ISR) \; , \\ 
\pb_o^{(ISS)}(RSS) & = & \pb^{(ISS)}(RSS) \; , \\ 
\pb_o^{(ISS)}(RSI) & = & \pb^{(ISS)}(RSI) \; , \\ 
\pb_o^{(ISS)}(RSR) & = & \pb^{(ISS)}(RSR) \; , \\ 
\pb_o^{(ISI)}(ISI) & = & \pb^{(ISI)}(ISI) \; , \\ 
\pb_o^{(ISI)}(ISR) & = & \pb^{(ISI)}(ISR) \; , \\ 
\pb_o^{(ISI)}(RSI) & = & \pb^{(ISI)}(RSI) \; , \\ 
\pb_o^{(ISI)}(RSR) & = & \pb^{(ISI)}(RSR) \; , \\ 
\pb_o^{(ISR)}(ISR) & = & \pb^{(ISR)}(ISR) \; , \\ 
\pb_o^{(ISR)}(RSR) & = & \pb^{(ISR)}(RSR) \; ,
 \end{array}
\end{equation} 
as well as
\begin{equation}\label{MarkMEworksI}
\begin{array}{rcl}
\pb_o^{(SIS)}(SIS) & = & \pb^{(SIS)}(SIS) \; , \\ 
\pb_o^{(SIS)}(IIS) & = & \pb^{(SIS)}(IIS) \; , \\ 
\pb_o^{(SIS)}(SII) & = & \pb^{(SIS)}(SII) \; , \\ 
\pb_o^{(SIS)}(III) & = & \pb^{(SIS)}(III) \;.
\end{array}
\end{equation} 
Intuitively, the results in \eqref{MarkMEworksI} hold because the intermediate case has not recovered yet: for any time $t$ at which the closure is studied, we know that the intermediate infective has been infectious for a non-random duration $t$ and therefore the events of infecting either of the neighbours are independent of each other. On the other hand, the results listed in
\eqref{MarkMEworksS} hold because the intermediate susceptible rules out the
presence of any correlation between the two extremes: either 1 cannot infect 3
(and therefore, e.g.~$\pb^{(ISS)}(ISI) = 0$) or, if both are infectious, their
behaviour is uncorrelated because 2 escapes infection independently from both. Again, both these qualitative behaviours transcend the Markovian model itself. 
More formally, focusing our attention only on states $ISS$ and $ISI$, we highlight the following result.
\begin{prop}\label{thm:SIRtree}
On an open triplet,
\begin{equation}
\label{SIRtree}
\pb_o(ISS) = \pb(ISS)\quad\text{and}\quad\pb_o(ISI) = \pb(ISI)
\end{equation}
for all times $t\ge 0$, when the initial conditions are $\bx^0=(ISS)$ or $\bx^0=(ISI)$.
\end{prop}
\begin{proof}
Consider first the case $\bx^0=(ISS)$. Clearly $\pb(ISI)=0$, so $\pb_o(ISI)=\pb(ISI)$ holds trivially. For state $\bx=(ISS)$ instead,
\[\pb_o(ISS)=\frac{\pb_{12}(IS)\pb_{23}(SS)}{\pb_2(S)}\;.
\]
Because the triplet is open, $\pb_2(S) = \pb_{23}(SS)$ and $\pb_{12}(IS) = \pb(ISS)$, and hence $\pb_o(ISS)=\pb(ISS)$.

Consider now $\bx^0=(ISI)$. Clearly, $\pb(ISS)=0$, as infected nodes cannot recover, so $\pb_o(ISS)=\pb(ISS)$ holds trivially. For state $\bx=(ISI)$ instead,
\[\pb_o(ISI)=\frac{\pb_{12}(IS)\pb_{23}(SI)}{\pb_2(S)}\;,
\]
but given the initial conditions all factors equal $\pb(ISI)$ and the closure is still exact.
\end{proof}

This result works for any assumptions about the infectious period, but is particularly important in the Markovian case because it provides the basis for why the pair-based approximation for a Markovian SIR epidemic spreading on a more general unclustered
network, in which only the $ISS$ and $ISI$ states appear, is exact (as long as the starting configuration is pure; see \cite{Sharkey:2013}). A formal proof of this first appeared in \citet{Sharkey:2013}. However, in Section \ref{sec:LargeTree} we will provide a simpler and more general proof.

\section{Constant infectious period: SIR model on three nodes}
\label{sec:Const}

Although the standard approximation $o$ on the open triplet is always exact, at least for the dynamically important states $ISS$ and $ISI$, we have argued that this is not the case for many other states because of the random duration of the infectious period.
We now show that the
random duration of the infectious period is the main reason why all
moment closure approximations fail to be exact, both on the open triplet and on the closed triangle. In what follows we assume a constant infectious period of duration $T\equiv\meanT$ and we arbitrarily assume that an individual is still infectious at $t=\meanT$ and is immune immediately after, i.e.~for $t>\meanT$.

\begin{prop}\label{thm:OpenMEConst}
For the SIR model on an open triplet, when the infectious period has constant duration,
\begin{equation}\label{OpenMEConst}
\pb_o(ABC) = \pb(ABC) \; ,
\end{equation}
for all $A,B,C\in\{S,I,R\}$, all times $t\geq 0$, and all initial conditions.
\end{prop}

\begin{proof}
We consider each initial condition separately.

\noindent (i) Assume $\bx^0 = (ISS)$ and recall \eqref{MEopen}. Because of the
initial condition, the probability of all cases in which the state $A$ of
individual 1 is $S$ is 0, i.e.~
\[\pb(SBC) = \pb_{12}(SB) = \pb_{13}(SC) = \pb_1(S) = 0\]
for all $B,C\in\{S,I,R\}$. Now consider separately the cases in
which $A=I$ and $A=R$.

When $A=I$, $\pb(ABC) = 0$ for all $t>\meanT$, but also $\pb_{12}(AB)$ and so $\pb_o
(ABC)$ are null and \eqref{OpenMEConst} holds trivially. Therefore, consider
only the times $t\leq \meanT$. Then, $\pb_{12}(RB) = \pb(RBC) = 0$, so that
\[\pb_{23}(BC) = \pb(SBC) + \pb(IBC) + \pb(RBC) = \pb(IBC) = \pb(ABC)\]
and
\[\pb_2(B) = \pb_{12}(SB) + \pb_{12}(IB) + \pb_{12}(RB) = \pb_{12}(IB) = P_{12}(AB),\]
and thus $\pb_o (ABC) = \pb(ABC)$. Therefore \eqref{OpenMEConst}
holds for all $t\geq 0$.

When $A=R$, the argument is similar. In particular, $\pb(ABC) = 0$ for all
$t\leq \meanT$, but also $\pb_{12}(AB)$ and therefore $\pb_o (ABC)$ are null. For
$t>\meanT$, instead, $\pb_{12}(IB) = \pb(IBC) = 0$, so that
\[\pb_{23}(BC) = \pb(SBC) + \pb(IBC) + \pb(RBC) = \pb(RBC) = \pb(ABC)\]
and
\[\pb_2(B) = \pb_{12}(SB) + \pb_{12}(IB) + \pb_{12}(RB) = \pb_{12}(RB) = P_{12}(AB),\] and thus $\pb_o (ABC)
= \pb(ABC)$. Therefore, again, \eqref{OpenMEConst} holds for all $t\geq 0$.

\noindent (ii) For all other initial conditions $\bx^0$ in which individual 1 is
non-susceptible at the start a similar argument to the one above can be used to prove that the proposition still holds.

\noindent (iii) When the initial condition is $\bx^0 = (SSI)$ or any other in
which individual 3 is non-susceptible from the start, the proposition also follows by symmetry.

\noindent (iv) For $\bx^0 = (SIS)$, the problem is slightly different. First of
all we take the standard convention that $\pb_o (ABC) = 0$ for $B=S$ (i.e.~we assume that a ratio is null when the numerator is null, irrespective of the value of the denominator). Then, for $B=I$ or $B=R$,
the result holds trivially when $t>\meanT$ or $t\leq \meanT$, respectively. When the
result is not trivial, it holds because for a constant duration of the
infectious period, individual 2 transmits (or has transmitted) independently to
1 and 3, so that the joint distribution of the state of pairs $(1,2)$ and
$(2,3)$ is the product of the marginals.

\noindent (v) All other initial states in which individual 2 is not susceptible
at the start follow trivially. \end{proof}

\begin{prop}\label{thm:ClosedMEConst}
For the SIR model on a closed triangle, when the infectious period has constant duration, all
moment closure approximations considered here are exact, i.e.
\begin{equation}\label{ClosedMEConst}
\pb_\kappa(ABC) = \pb_\rho(ABC) = \pb_\mu(ABC) = \pb(ABC)\; ,
\end{equation}
for all $A,B,C\in\{S,I,R\}$, all times $t\geq 0$, and all initial conditions.
\end{prop}
\begin{proof}
We analyse each moment-closure approximation separately.
\medskip

\noindent \textbf{(a) Kirkwood.} For the case of Kirkwood's approximation
$\kappa$, we need to prove that:
\begin{equation} \label{stappConst}
\frac{\pb_{12}(AB) \pb_{23}(BC) \pb_{13}(AC)}{\pb_1(A) \pb_2(B)\pb_3(C)}
= \pb(ABC)\; . 
\end{equation}

\noindent (i) Consider first the initial condition $\bx^0 = (ISS)$. Because the
initial condition, the probability of all cases in which the state $A$ of
individual 1 is $S$ is 0, i.e.~
\[\pb(SBC) = \pb_{12}(SB) = \pb_{13}(SC) =
\pb_1(S) = 0\] for all $B,C\in\{S,I,R\}$. Now consider separately the cases in
which $A=I$ and $A=R$.

When $A=I$, $\pb(ABC) = 0$ for all $t>\meanT$, but also $\pb_{12}(AB)$ and so
$\pb_\mu (ABC)$ are null and \eqref{stappConst} holds trivially (we adopted the convention that the indeterminate form $0/0$ equals $0$). Therefore,
consider only the times $t\leq \meanT$. Then, $\pb_{12}(RB) = \pb_{13}(RC) =
\pb(RBC) = 0$, so that \[\pb_2(B) = \pb_{12}(SB) + \pb_{12}(IB) + \pb_{12}(RB) =
\pb_{12}(IB) = P_{12}(AB)\] and, similarly, $\pb_3(C) P_{13}(AC)$. Therefore,
\[ \pb_\sigma(ABC) = \frac{\pb_{23}(BC)}{\pb_1(A)} = \pb(ABC)\; , \]
because \[\pb_{23}(BC) = \pb(SBC) + \pb(IBC) + \pb(RBC) = \pb(IBC) = \pb(ABC)\]
and $\pb_1(A)=1$. Hence, \eqref{stappConst} holds for all $t\geq 0$.

\noindent (ii) A similar argument holds for $A=R$, but now the trivial case
when all probabilities are 0 is when $t\leq \meanT$ and the other considerations
apply to $t>\meanT$.

\noindent (iii) The calculations for any other initial condition in which
individual 1 is non-susceptible are simply a special case of those given above.

\noindent (iv) The result works for any other initial condition for symmetry
reasons (we can simply define individual 1 to be an initial non-susceptible).

\noindent \textbf{(b) First-step ME.} For the case of approximation $\rho$,
obtained by stopping the ME algorithm after a single triple step, we need to
prove that:
\begin{equation}\label{MEapp}
\frac{\pb_{12}(AB) \pb_{23}(BC) \pb_{13}(AC)}{\pb_2(B)
\sum_{b}{\frac{\pb_{12}(Ab)\pb_{23}(bC)}{\pb_2(b)}}} = \pb(ABC) \; .
\end{equation}

\noindent (i) Analogously to before, first assume we start from $\bx^0 = (ISS)$
and consider the case $A=I$ and $t\leq \meanT$ (the result is trivial for $t>\meanT$).
Then: \[ \pb_{\rho}(ABC) = \frac{\pb(ABC)\pb_{13}(BC)}{\sum_{b}\pb(AbC)} =
\pb(ABC)\; . \] The case $A=R$ and $t>\meanT$ is analogous.

\noindent (ii) Given the asymmetry in \eqref{MEapp}, we still need to consider
separately the case of $\bx^0 = (SIS)$ (all the others work as special cases
by symmetry).  In this case, for $t\leq \meanT$, the sum at the denominator of
\eqref{MEapp} contains only the term for $b=I$, and cancels out with the first
2 terms at the numerator ($\pb_2(b)=1$). The result follow immediately because
$\pb_2(B)=1$ and $\pb_{13}(AC) = \pb(ABC)$.

\noindent \textbf{(c) Full ME.} Finally we prove the result for the ME
approximation $\mu$. We already know from \textbf{(b)} above that if
we start the ME algorithm from the uniform distribution $\pb^0 (ABC) = 1/27$,
for all $A,B, C \in\{S,I,R\}$, after the first triple step we reach
$\pb^{(1)}(ABC) = \pb_{\rho}(ABC) = \pb(ABC)$. We now show that, if we apply
another triple step to $\pb(ABC)$, we remain on the same distribution
$\pb(ABC)$. In other words, starting from any initial distribution, convergence
of the ME algorithm, restricted to its output after every triple step, occurs after a single triple step.

To show this, we expand the first triple step of the algorithm (as in
\eqref{Rhoapp}, but by keeping explicitly $\pb^{(0)}(ABC)$ in the equations).
Considering separately every initial condition with one infective and 2
susceptibles, we use the same arguments used in other proofs (for example, when
$\bx^0 = (ISS)$, in the suitable time range, we know that $\pb_{23}(BC) =
\pb(ABC)$ and that sums over the state $a$ of individual 1 contain only the
probabilities for $a=A$) to prove that $\pb^{(1)}(ABC) = \pb(ABC)$.
\end{proof}

\begin{rem}
It is worth mentioning that the proofs of Propositions \ref{thm:OpenMEConst} and \ref{thm:ClosedMEConst} above require only the
initial infective (or any of the initial infectives, if more than one) to have
a constant duration of the infectious period. Therefore the result readily
extends to the case in which the three individuals have possibly different
durations of infection, as long as they are
non-random. Furthermore, in the case of the $SI$ model, the same
arguments can be used to prove that all approximations are always exact.
\end{rem}

\section{Erlang-distributed infectious period: SIR model on three nodes}
\label{sec:GammaDistributedInfectiousPeriod}

Analytical progress becomes difficult when the infectious period is not constant. Conversely, numerical methods based on continuous-time Markov chains are straightforward to implement to study the behaviour of the Markovian model. 
In order to bridge the gap between these two extremes, we extend the framework
used for Markovian model by allowing infectives to go through a series of
infectious stages, each with independently and identically distributed
exponential infectious periods, to model an overall Erlang-distributed sojourn
time in the infectious state, with mean $\meanT=1$. As the number $n_I$ of infectious
stages increases, the variance decreases as $1/n_I$. The results are reported in
Figure \ref{fig:OpenSSDVSvar}: for each of the three starting points $\bx^0$
(at time $t=0$), the overall ``distance'' between the exact distribution and
the moment-closure approximation is measured by taking, at each time $t\geq 0$,
the sum of squared differences (SSD), $\sum_{\bx}{[\:\pb^{\bx^0}(\bx;t) -
\pb_o^{\bx^0}(\bx;t)\:]^2} $, and then integrating it over all times. The choice
of SSD to measure the discrepancy between distributions instead of
the possibly more natural Kullback-Leibler (KL) divergence is appropriate because
Kirkwood's approximation does not lead to a proper distribution over system
states, thus often yielding negative values for the
KL divergence that are hard to interpret.

\subsection{Open triplet}
\label{sec:GammaOpen}

Figure~\ref{fig:OpenSSDVSvar} highlights how the exact result for the
theoretical limit of a constant duration of infection is approached and how
such convergence depends on the starting state $\bx^0$ and the infection rate
$\tau$. Note in particular how the slowest convergence seems to be attained for
intermediate (though starting-state-dependent) values of $\tau$ and also how
the approximation performs particularly poorly for the starting condition
$\bx^0=(SIS)$.

Figure~\ref{fig:OpenSSDVSvar} combines the contribution of all states in an
aggregate measure of the approximation performance. Decomposing such an
aggregate measure, however, reveals significant heterogeneity (see Figures S1
and S2 of the Supplementary Material), with negative and positive errors in
different cases. The decomposition, however, confirms the presence of cases
when the approximation is exact and the particularly poor performance when the
system starts with the intermediate node infected (bottom row of Figures S1 and
S2 of the Supplementary Material).

\subsection{Closed triangle}
\label{sec:GammaClosed}

On the closed triangle, Figure~\ref{fig:ClosedSSDVSvar} describes the overall
``distance'', measured again by the time-integrated SSD, between each
approximation and the exact probability distribution over all systems states,
for different starting states, as a function of the variance in the duration of
the infectious period. On the other hand, Figure~\ref{fig:ClosedSSDVStau}
expresses the same distance as a function of the infectivity $\tau$. It is
quite evident that, overall, ME performs better than the other approximations,
in particular Kirkwood's. However, the rate of convergence seems to be
comparable to that of other closures.

Again, assessing the overall quality of each approximation on the closed
triangle with such an aggregate measure hides the strong heterogeneity that can
be seen when decomposing the SSD in the contribution of different states. In
general (see Figures S3-S6 of the Supplementary Material), ME performs often
better than Kirkwood's approximation, although we collect in Figure
S7 some of the extreme examples of its variability
in performance. In particular, in the somewhat trivial case of $(ISS)\to(ISS)$,
we have found ME to be consistently less accurate than Kirkwood's approximation
over the entire parameters' spectrum, sometimes by almost an order of
magnitude. The implications of such heterogeneity on the performance of each
approximation in any practical context are still unclear and require further
investigation. In fact, despite ME appearing superior when the contribution of
all states is equally weighted, when studying the system dynamics at the
population level, the infection process, the current epidemic phase, as well as
the specific network structure, all interact in a complex fashion to produce
different proportions of triangles in each state. Uneven weights associated to
each of the specific transitions $\bx^0\to\bx$ could overturn the current
conclusions in some cases. In particular, it is not unreasonable to imagine a
large proportion of triangles that have started from but not yet left
the $(ISS)$ state, a case in which ME performs particularly poorly.

A final comparison between the various approximation methods on the closed
triangle consists in stratifying the contribution of each state $\bx$ to the
overall SSD measure (see Figure S8 in the Supplementary Material). In addition to
the quantitatively smaller error of the ME approximation compared to the
others, its evidently most balanced decomposition both across states and in
particular over time undoubtedly represents a further element of merit.

\section{Results for other motifs}
\label{sec:OtherMotifs}
In Sections \ref{sec:Markovian}, \ref{sec:Const} and \ref{sec:GammaDistributedInfectiousPeriod} we have focused on the SIR model
mainly because it is the one that is most commonly considered in the literature. 
We have also restricted our attention to a single open triplet or closed 
triangle, because they are simple enough for analytical results to be obtained. 
We now show that some of these results are peculiar to the particularly simple 
topologies considered, and in particular on the fact that the presence of 
initially infected cases \emph{inside} the triplet or triangle significantly 
reduces the degrees of freedom of the system. However, such initial careful 
analysis provides important insight, which guides us in how to approach 
extensions to other motifs and to larger networks, in particular in terms of the 
role that a random duration of the infectious period plays on the exactness of 
various moment closures.


Analytical proofs become cumbersome as the complexity of the graph increases, and 
appears particularly difficult to obtain in the case of ME. Therefore, we opt for 
extending the numerical exploration of Section 
\ref{sec:GammaDistributedInfectiousPeriod} to build some intuition about the 
errors in moment closure approximations on open triplets and triangles that are 
subgraphs of other slightly larger networks. We note from the start that, if the 
aim is to understand how our results extend to large networks, at some point the 
numerical method we are using needs to be abandoned in favour of a dynamical 
system where one can can deduce the rate of change of the state of a pair based 
on the states of its neighbours \citep{Sharkey:2013}.

Although recent results \citep{Karrer:2010,Wilkinson:2014} involve the formulation of dynamical 
systems based on time-since-infection approaches, leading to equations with 
distributed delays, by far the easiest starting point for writing a dynamical 
system for each pair of nodes is to use ordinary differential equations. This 
requires the use of constant rates (i.e.~Markovian models). However, in 
Proposition \ref{thm:ClosedMEConst} we have shown 
that, as soon as we introduce randomness in the duration of the infectious 
period, moment closure approximations fail on networks with loops.

The SI model, apart from being simpler because of the presence of only 2 states, has the further benefit of being both Markovian and with a ``constant'' (infinite) duration of infection. Therefore, we focus on it as our baseline model for exploring the exactness of moment closure approximations for networks with more than three nodes, with the understanding that closures that fail in the SI model cannot be exact if latency or recovery are added.


Figure \ref{fig:graphs} reports a range of motifs of increasing complexity. The three digits appended to the end of the motif name are the indices of the nodes (as in Figure \ref{fig:graphs}) of the triplet on which the closure is applied. Unless stated otherwise, the initial condition is represented by node 1 having just been infected.

\subsection{SI model}
\label{sec:OtherMotifsSI}

Table \ref{tab:SI} reports a comprehensive list of motifs, based on Figure \ref{fig:graphs} on which the moment closures are tested for the SI model. Because of the high dimensionality of the exploration, we only report the two dynamically important states $(ISS)$ and $(ISI)$ and, to further enrich our understanding, states $(IIS)$ and $(III)$. We choose to observe the error for those four state at only one time point per motif, for which we have checked that the probability of the system being in the most interesting states is non-zero. Unless specified otherwise, such time is taken to be $t=0.5$ (to be compared with the time scale $\meanT=1$, which we use in the presence of recovery). 

Approximations are divided in three groups, where the first and second assume one initially infected node and, respectively, the closure on an open triplet or a closed triangle, while the third assumes multiple initially infected nodes. In each group, the list of cases examined is further grouped in subsets with the intention of testing a particular network feature over small graphs of increasing complexity. Table \ref{tab:SI} indicates whether each closure approximation is exact or not at the time tested (hereafter we will say it ``works''), although wherever possible the specific time is chosen such that a positive answer is indicative of general validity for all times (a negative answer, or ``failure'', is of course sufficient to discard exactness). Only ME is tested on closed triangles, because if it fails, Kirkwood's approximation and 1-step ME also fail. Comments are added to provide extra information where appropriate. Areas of grey background are the most useful ones to gain understanding of whether the results of Propositions \ref{thm:OpenMEConst} and \ref{thm:ClosedMEConst} scale up to larger networks and up to which point.

\subsubsection{Open triplet}
\label{sec:OtherMotifsSIopen}

We know from Section \ref{sec:Markovian} and from Proposition \ref{thm:OpenMEConst} that the SI model works on the open triplet with a single initial infective. We verify that extensions to any tree-like structure also work, in line with the results of \citet{Sharkey:2013} and Theorem \ref{thm:TreeAllModels} below. Triplets where the infection enters from the central node (3star324 and Tree324) do not contribute to the dynamics of the spread on a tree.

Although the closure appears to work on loops of size 4 when the system starts inside the loop (Square123 and toastB123), if the system starts outside the loop (KiteEmpty234 and FishEmpty345, or KiteDiagB234 and FishDiagB345) the closure fails, thus suggesting none of the closures examined here extends to large networks containing loops larger than 3. The behaviour is slightly different if the system starts already in the loop but not in the triplet (Square234 and ToastA234). However, starting again outside the loop and entering the loop not in the triplet (KiteEmpty345, FishEmpty456, KiteDiagA345 and FishDiagA456) the closure fails, confirming the impossibility to scale our exact results on large networks with loops of more that 3 nodes.

\subsubsection{Closed triangle}
\label{sec:OtherMotifsSIclosed}

Proposition \ref{thm:ClosedMEConst} guarantees that all closures are exact on a closed triangle. This appears to be true even for slightly more complex networks (ToastA123 and 4Full123) as long as the system starts in the triangle on which the closure is applied. ME seems to be exact also on closed triangles when the system does not start within the triangle (MartiniGlass234 and BowTie345) and the same behaviour applies when the triangle is part of a larger motif (KiteDiagA234 and FishDiagA345, or KiteFull234 and FishFull345). Note however that, whenever the system starts outside the triangle, both Kirkwood and 1-step ME fail. This suggest that neither of these closures scales up to larger networks with triangles and only ME gives hopes for such an extension. However, as soon as the system is allowed to enter the triangle on which the closure is applied through more than one route (ToastB234 and all the following networks) even ME fails. 

We further confirm our intuition that ME is exact even in the presence of overlapping triangles, on condition that the each triangle can only be entered through a single route, by observing that ME works also on two special larger networks, namely DoubleKite237 and DoubleFish348.

Note that, although ME works on the first triangle encountered  in a full clique of size 4 (4Full123, KiteFull234 and FishFull345), it does not work on other triangles inside the same clique (4Full234, KiteFull345 and FishFull456), therefore leaving no hope for any of the closures investigated here to work on networks containing fully connected cliques of size larger than 3. In particular, this suggests that, in the so-called \emph{households models}, where fully connected cliques (households) are joined by a few between-clique links, the dynamics of infection spread on single nodes or pairs (and hence, for example, the expected epidemic course) cannot be represented exactly only in term of pairs using any of the moment closure techniques considered here, unless households have size no larger than 3.

\subsubsection{Multiple initial infectives}
\label{sec:OtherMotifsSImultiple}

In line with \citet{Sharkey:2013} and Theorem \ref{thm:TreeAllModels} below, we verify that the the closure on the open triplet works on to tree-like networks also with multiple initial infectives, on condition that the initial conditions are pure.

However, on closed triangles, we verify that even ME fails when multiple initial infectives are present, even if the initial conditions are pure (Tripod234). This suggests that even when the maximum household size is 3, the moment closures considered here are exact only when a single initial case starts the epidemic.

\subsection{SIR model}
\label{sec:OtherMotifsSIR}

In Proposition \ref{thm:ClosedMEConst} and Section \ref{sec:GammaDistributedInfectiousPeriod} we have shown that, if we introduce the possibility of recovering, exactness of any closure on a network that contains some triangles can only be hoped for with a constant duration of infectious period. Therefore, in Table \ref{tab:SIR} we report a similar analysis to the one done in Table \ref{tab:SI}, but for the SIR model with constant duration of the infectious period $T\equiv\meanT$ (identified by the letter ``C''). The results are not exact but are inferred by visually examining the convergence, like in Figure \ref{fig:ClosedSSDVSvar}, as the number $n_I$ of infectious classes increases. The potential exactness of the approximation also for the Markovian model with exponentially distributed $T$ (denoted by ``M''), and therefore for all models with $T$ random but not-degenerate, is reported in the comments.

For ease of comparison with the careful examination presented for the SI model, in Table \ref{tab:SIR} we propose the same structure of Table \ref{tab:SI}. However, we know that, if on a particular graph the closure fails for the SI model, it is bound to fail also for the SIR model. Therefore, we only fill in the table partially, leaving aside all tests that do not help gaining further understanding.

On the open triplet, Table \ref{tab:SIR} confirms the exactness of the closure on  tree-like structures for any infectious period. It also confirms that the closure is not exact on loops of size 4, if the initially infected node is outside the loop, thus suggesting no large network with loops larger than triangles admit exact dynamics under the pair-based approximations considered.

However, the results for the MartiniGlass234 and the KiteDiagA234 suggests that there is no hope for the studied approximations to be exact even if loops only consist of closed triangles, whenever the initially infected node is outside the triangle and even if the infectious period is of constant duration. Therefore, we conclude that, for SIR models, the pair approximations considered can only be exact in general on a tree-like network (see \citet{Sharkey:2013} and Theorem \ref{thm:TreeAllModels} below).

As an example, we show two results of our numerical investigation, one suggesting exactness on the triangle (Figure \ref{fig:TriangleSIR}) and one suggesting failure to be exact on the MartiniGlass234 graph (Figure \ref{fig:MartiniGlass234_SIR}).

\subsection{SEIR}
\label{sec:OtherMotifsSEIR}

Given the negative result obtained for the SIR model, we do not expect the SEIR to perform any better. For this reason we only report our results in Table S1 in the supplementary material. Results are mostly for the SEIR model with constant duration of both the latent and the infectious period, hereafter denoted by ``CC'', though comments on cases where the latent, the infectious, or both are Markovian (MC, CM and MM, respectively) are given when useful. Again, results are not exact, but rather extrapolations from trends for increasing number of the latent classes, $n_E$, and of infectious classes $n_I$. Given most models in the literature only consider the SIR model, it is interesting to verify, in line with Theorem \ref{thm:TreeAllModels} below, that the fact that pair approximations work on a tree-like structure extends to the additional presence of a latent period.

The exploration is computationally intensive, given the number of classes involved. Therefore, some cases were dubious and we did not feel we could conclude anything with confidence. However, they do not affect the whole picture.

\subsection{Reed-Frost-type models}
\label{sec:OtherMotifsRF}

The models of Reed-Frost type (RF) represent a special case that needs to be treated with care. In particular, nobody is ever in the I state. We have visually explored moment closure on many triplet and triangle states, but we here present only states ESS and ESE, under the assumption that they are the dynamically important ones, as well as EES and EEE, to keep the parallel with the previous cases. Exploration of other states did not contribute in gaining further insight. The initially infected node (or nodes) are assumed to have just entered state E.

As noted in Section \ref{sec:EpiModelsRF}, RF-type models can be studied numerically by considering an SEIR model and assuming an infectious period that is much shorter than the latent period. In all numerical analyses, we have used $\meanL = 1$ for the latent period and $\meanT = 0.001$ for the infectious one. The infection rate is adjusted to $\tau = 1000$ to keep fixed the mean probability of transmitting across a link, $\tau\meanT=1$. Furthermore, we investigated both the ``standard'' RF model with a constant duration of infectious period $L\equiv\meanT$ and a fixed probability $P\equiv p$, denoted by ``CC'' and approximated by letting both $L$ and $T$ being Erlang-distributed with decreasing variance (by increasing the number of stages $n_E$ and $n_I$), as well as the cases, denoted by ``MC'', ``CM'' and ``MM'', where either $L$ or $T$ or both are exponentially distributed, respectively. 

Results are similar to those of the SI model: in particular, we found that ME seems to be exact also on closed triangles even when the system does not start within the triangle (MartiniGlass234 and BowTie345). The same behaviour applies when the triangle is part of a larger motif (KiteDiagA234 and FishDiagA345, or KiteFull234 and FishFull345), as long as the triangle can be accessed only through one route. Unlike the SI model, here 1-step ME also appears to be exact in certain cases, though Kirkwood's approximation still fails. Surprisingly, we found that most results that hold for the CC case also hold for a random latent period and a random transmission probability $P$.

As for the SEIR model, the exploration in this case is computationally intensive, given the number of classes involved and, in addition, the numerical challenges of having both small and very large rates simultaneously. As before, dubious cases are highlighted, but do not affect the whole picture. Unlike the SEIR model, however, the exploration can be somewhat simplified by noting that, in the RF-CC model, many states never occur with positive probability. We carefully selected the times when to investigate each closure, and monitored also the probability with which the system can be in those states at those times, to make sure results were not trivial. Figure \ref{fig:KiteDiagB345} reports an example of the numerical exploration in the dubious case of the KiteDiagB345 graph. Despite the open possibility that convergence for states ESS and ESE might occur if more classes could be added (we believe it not to, though), Figure \ref{fig:KiteFull345_RF} for the KiteFull345 graph clearly shows the error increasing, strongly suggesting the approximation is unlikely to be exact in general anyway.

\section{Extension to large networks}
\label{sec:LargeNetworks}

\subsection{Tree-like networks}
\label{sec:LargeTree}

In Proposition \ref{thm:SIRtree} we showed that moment closure on the dynamically important states $ISS$ and $ISI$ for the SIR model on an open triplet is exact. Our numerical exploration in Section \ref{sec:OtherMotifs} suggests it holds for larger networks and different models and Figure \ref{fig:sitrees}b confirms it via simulation for the SI model with a single initial infective. In line with the above, \citet{Sharkey:2013} proved that the same results hold for the Markovian SIR model on any tree-like network and any number of initial infectives, as long as the initial condition is pure.

Here we provide a much simpler and more intuitive proof of this results, which holds much more generally (in particular for all models considered here).

\begin{thm}
\label{thm:TreeAllModels}
For any connected triplet $(i,j,k)$ on any tree-like network, and any model of infection spread in which susceptibles escape infection independently from each of their infected neighbours and, after infection, can never return to the susceptible state,
\begin{equation}
\label{TreeAllModels}
\pb^{\bx^0}_{i\!jk,o}(ISS) = \pb^{\bx^0}_{i\!jk}(ISS) \quad \text{and} \quad \pb^{\bx^0}_{i\!jk,o}(ISI) = \pb^{\bx^0}_{i\!jk}(ISI) \; ,
\end{equation}
for any $t\ge 0$ and any pure initial condition $\bx^0$. The result should be adapted to the Reed-Frost-type models by replacing $I$ with $E$ in \eqref{TreeAllModels}.
\end{thm}

\begin{proof}
We initially provide the shortest proof, for which we only need statements involving nodes $i,j$ and $k$. This is fully general, but we believe that explicitly showing how the nodes of the triplet interact with the neighbouring nodes might clarify the argument even further. Therefore, we later present a slightly longer elaboration, applied to the particular case of the Vine246 (Figure \ref{fig:graphs}). 

Consider any pure initial condition $\bx^0$, and use the following notation to describe events: 
\begin{itemize}
\item[$S_j$]: node $j$ is susceptible at $t$;
\item[$I_{t_i}$]: node $i$ has been infected at time $t_i<t$ and is still infectious at time $t$.
\end{itemize}
Note that, although not explicitly stated, both these events depend on $\bx^0$. Then
\begin{equation}
\label{Pij}
\pb_{i\!j}(IS) = \int_0^t{\pb\pp{S_j \wedge I_{t_i}}\rd t_i} = \int_0^t{\pb\pp{S_j \cond I_{t_i}} \pb\pp{I_{t_i}}\rd t_i}
\end{equation}
and 
\begin{equation}
\label{Pjk}
\pb_{jk}(SS) = \pb\pp{S_j \cond S_k} \pb\pp{S_k}.
\end{equation}
Also,
\begin{eqnarray}
\pb_{i\!jk}(ISS) & = & \int_0^t{\pb\pp{I_{t_i} \wedge S_j \wedge S_k}\rd t_i} \label{Pijk_1}\\
& = & \int_0^t{\pb\pp{I_{t_i} \wedge S_j \cond S_k} \pb\pp{S_k}\rd t_i} \label{Pijk_2}\\
& = & \int_0^t{\pb\pp{I_{t_i} \cond S_j} \pb\pp{S_k \cond S_j} \pb\pp{S_j}\rd t_i} \label{Pijk_3} \\
& = & \int_0^t{\frac{\pb\pp{S_j \cond I_{t_i}} \pb\pp{I_{t_i}}}{\pb\pp{S_j}} \frac{\pb\pp{S_j \cond S_k}\pb\pp{S_k}}{\pb\pp{S_j}}  \pb\pp{S_j}\rd t_i} \label{Pijk_4} \\
& = & \frac{\lsq\int_0^t{\pb\pp{S_j \cond I_{t_i}}\pb\pp{I_{t_i}}\rd t_i}\rsq}{\pb\pp{S_j}} \lsq\pb\pp{S_j \cond S_k}\pb\pp{S_k}\rsq \label{Pijk_5} \\
& = & \frac{\pb_{i\!j}(IS) \pb_{jk}(SS)}{\pb_j(S)} \label{Pijk_6}\\
& = & \pb_{i\!jk,o}(ISS). \label{Pijk_7}
\end{eqnarray}
Here, the key passage is between \eqref{Pijk_2} and \eqref{Pijk_3}: conditional on node $j$ being susceptible, the states of nodes $i$ and $k$ are independent. This heavily relies on the tree-like structure and the assumption that individuals cannot regain susceptible status after having been infected, so that if node $j$ is susceptible at time $t$, it has been so for all times from 0 to $t$, and this has prevented any information from passing from node $i$ to $k$ or vice versa. The step from \eqref{Pijk_3} and \eqref{Pijk_4} follows directly from the definition of conditional probability.

The $ISI$ case follows the same steps, though the behaviour of pair $(j,k)$ mirrors that of $(i,j)$ and the proof involves a double integral over both $t_i$ and $t_k$.
\end{proof}
\begin{proof}[Proof applied to the Vine246 case]

We now work out a slightly more laborious proof, where we explicitly consider the neighbouring nodes of the triplet. For this we use the notation relative to the Vine246 network (Figure \ref{fig:graphs}), although it is clear that generalisation to any tree-like network is straightforward. We now write $H(i_1,i_2,\dots,i_n)$ to denote the joint history of the states $i_1,i_2,\dots,i_n$ in the time interval $[0,t]$, and we denote generically by $\int_{H(i_1,i_2,\dots,i_n)}$ the integral over all possible such joint histories. The passages of the proof are essentially the same as above, so we only focus on the $ISI$ state.

We have
\begin{align}
& \pb_{246}(ISI) \quad = \\
& = \; \int_{H(1,3,5,7,8)}{\int_0^t{\int_0^t{\pb\pp{S_4 \wedge I_{t_2} \wedge I_{t_6} \wedge H(1,3,5,7,8)}\rd t_2}\rd t_6}} \label{P246_1}\\
& = \; \int_{H(1,3,5,7,8)}{\int_0^t{\int_0^t{\pb\pp{I_{t_2} \wedge I_{t_6} \wedge H(1,3,5,7,8) \cond S_4} \pb(S_4) \rd t_2}\rd t_6}} \label{P246_2}\\
& = \; \int_{H(1,3,5,7,8)}{\int_0^t{\int_0^t{\pb\pp{I_{t_2} \wedge H(1,3) \cond S_4} \pb\pp{ I_{t_6} \wedge H(7,8) \cond S_4} \pb\pp{H(5)\cond S_4} \pb(S_4) \rd t_2}\rd t_6}}. \label{P246_3}
\end{align}
The last passage is the key step due to independence between all branches separated by node 4 (i.e.~rooted in nodes 2, 6 and 5). We now consider the separate factors inside the integral. Using the definition of conditional probability and the law of total probability,
\begin{align}
& \int_{H(1,3)}{\int_0^t{\pb\pp{I_{t_2} \wedge H(1,3) \cond S_4}\rd t_2 } } \;=\label{P246_H13_1} \\
& \qquad = \; \int_{H(1,3)}{\int_0^t{ \frac{\pb\pp{S_4 \wedge I_{t_2} \wedge H(1,3)} }{\pb\pp{S_4}} \rd t_2}} \label{P246_H13_2}
\\
& \qquad = \; \frac{\int_0^t{\int_{H(1,3)}{\pb\pp{S_4 \wedge I_{t_2} \wedge H(1,3)} \rd t_2}}}{\pb\pp{S_4}} \label{P246_H13_3}\\
& \qquad = \; \frac{\int_0^t{\pb\pp{S_4 \wedge I_{t_2}} \rd t_2}}{\pb\pp{S_4}} \label{P246_H13_4}\\
& \qquad = \; \frac{\pb_{24}(IS)}{\pb_4(S)}\; .\label{P246_H13_5}
\end{align}
The term involving event $I_{t_6}$ is dealt analogously. Instead, the term involving node 5, for which no information is available, simplifies to
\begin{align}
\label{P246_H5_1}
\int_{H(5)}{\pb\pp{H(5)\cond S_4}} & = \int_{H(5)}{\frac{\pb\pp{S_4\cond H(5)} \pb\pp{H(5)}}{\pb\pp{S_4}}} \\
 &= \frac{\int_{H(5)}{\pb\pp{S_4\cond H(5)} \pb\pp{H(5)}}}{\pb\pp{S_4}}  = 1. \label{P246_H5_2}
\end{align}
Substitution of \eqref{P246_H13_5}, its analogous for $I_{t_6}$ and \eqref{P246_H5_2} into \eqref{P246_3} leads to the desired result.
\end{proof}

\begin{rem}
Note that the statement used in the proof that information cannot pass through a susceptible node implicitly relies on the further assumption that the state of a node only depends on those of its neighbours and not on the neighbour's neighbours (nor on any other nodes). Some models might violate this assumption, although it becomes then questionable whether a tree-like static network is a good representation for such models.
\end{rem}

\begin{rem}
Also, note that in most epidemic models the infectious life of a newly infected node evolves as an autonomous process: i.e.~it is not affected by the neighbours or more generally by the environment. This assumption is convenient but not strictly necessary for Theorem \ref{thm:TreeAllModels} and can be relaxed: for example, one can imagine the states of nodes 1 and 3 in \eqref{P246_H13_3} affecting how and when node 2 progresses, say, from the latent to the infectious stage. 
\end{rem}

\begin{rem}
Further, one can even imagine nodes 1 and 3 affecting the probability that 4 remain susceptible up to time $t$. This is the case, for example, of node 1 being also connected to 4, i.e.~nodes 1, 2 and 4 forming a closed triangle. Then the result of Theorem \ref{thm:TreeAllModels} would not apply to the triplet (1,2,4), but it would still apply to triplets (1,4,6) and (2,4,6). More generally, the components of all branches stemming from node 4 need not be sub-trees. The only requirement is that the components containing the first and last node of the triplet on which the closure is applied do not communicate if node 4 is susceptible. Therefore, Theorem \ref{thm:TreeAllModels} can be extended to that case of triplets in which the central node is, in the terminology of \citet{Kiss:2014}, a \emph{cut-vertex}.
\end{rem}

\begin{rem}
Finally, note that the need for a pure initial condition comes from the fact that if the initial condition is random we would need to average both sides of the Equations in \eqref{TreeAllModels} over its distribution, thus getting, for example for state $ISS$,
\begin{equation}
\label{MixedIC1}
\pb_{246}(ISS) = \E{ \pb_{246}^{\bx^0}(ISS) } = \E{ \frac{\pb_{24}(IS) \pb_{46}(SS)}{\pb_4(S)} },
\end{equation}
which is in general different from
\begin{equation}
\label{MixedIC2}
\frac{\E{\pb_{24}(IS)} \E{\pb_{46}(SS)}}{\E{\pb_4(S)}} = \pb_{246,o}(ISS).
\end{equation}
\end{rem}

Before concluding this section, we further point out that, on a tree-like network the simplest local moment closure that retain exactness is the pair-based approximation. Consider in fact an open triplet with node 2 being the central node and assume that the initial state is $\bx^0=(ISS)$. The only natural closure for approximating the probabilities over pairs in terms of the probability of the states of single nodes is the one referred to as \emph{individual-based approximation} in \citet{Sharkey:2008}. Denoting it by $\pi$, it can be written, for example for the SI model and for pair $(2,3)$ in state $IS$, as $\pb_{23,\pi}(IS) = \pb_2(I) \pb_3(S)$. However, $\pb_{23,\pi}(IS) =\left[\pb(IIS)+\pb(III)\right] \left[\pb(ISS) +\pb(IIS)\right]$, which simple algebra shows is in general different from $\pb(IIS)$, so that the closure of the level of single nodes is not exact even in this simple case.

\subsection{Networks containing closed triangles}
\label{sec:LargeTriangles}

The examples of Section \ref{sec:OtherMotifs} suggest that, for the SI model starting with a single initial infective, Kirkwood's and 1-step ME approximations  fail everywhere, but ME seems to work for the dynamically important states $ISS$ and $ISI$ if the network contains triangles that do not overlap (MartiniGlass234 and BowTie345).

In Figure \ref{fig:sitrees}d we confirm this intuition by simulations on a
large network of non-overlapping triangles. The results are obtained by
numerically solving the set of ODEs, for the SI model, closed at the level of
pairs using Kirkwood and ME. Code for the former was provided by \citet{Sharkey:2011},
and we provide \textsc{Matlab} code for the latter as Electronic Supplementary Material. This
verifies that the latter, unlike the former, is able to capturing the dynamics
over non-overlapping triangles correctly, in the case of a single initially
infected node.

The discussion of Section
\ref{sec:OtherMotifsSIclosed} suggests that ME does work correctly
also in the presence of some overlapping triangles, as long as the infection is
not allowed to enter the same triangle through two different routes
simultaneously (e.g.~ToastB234); it can though enter, then leave and re-enter
(e.g.~KiteDiagA), as long as there is only one introduction point in the
triangle. However, we have also noticed in Section
\ref{sec:OtherMotifsSIopen} that the closure on an open triplet inside
a loop larger than a triangle is in general not exact: therefore, for example
on the KiteDiagA graph, while in the dynamics on the closed triangles 234 and
254 are handled correctly by ME, the closure applied on the open triplet 345
fails. The ODEs for a large network of ToastA motifs,
	however, are numerically unstable, meaning that we were not able to
	determine the exactness of the overall dynamics on large networks with
overlapping triangles like those of ToastA or KiteDiagA motifs,
which therefore remains an open problem.

\subsection{Further Intuition}
\label{sec:LargeIntuition}

The details of moment closure performance are complex, but
our intuition is that the two factors
that make closures fail are simply: (1) mixed initial conditions and (2)
random time at which recovery occurs (which can
be the case even with a constant duration of infection, if the time of
infection is random). The former is already known to create problems
\citep{Sharkey:2013}; the latter has been highlighted here in Proposition
\ref{thm:ClosedMEConst}.


The reason why all closures considered here fail for the SIR model on the MartiniGlass234 graph is due to factor (2): even with constant duration of infectious period, if we know that node 2 is in state  $I$ at $t$, we do not know when it recovers, as that depends on how long before $t$ it was infected. So, for any graph that contains a triangle but the initial condition is not in the triangle, all closures will fail when an individual can be in state I (i.e.~not RF) and recovers after a finite time (i.e.~not SI).

One last comment worth mentioning involves consideration for the absorbing states over triangles in the presence of recovery ($RSS$, $RRS$ and $RRR$) for $t\to\infty$ (see Figures S9 and S10 in the Supplementary Material). Unlike Kirkwood and 1-step ME, the closure based on ME seems to be able to capture the distribution over the absorbing states correctly for the MartiniGlass234 (Figure S9), even in the case of the SIR and the SEIR models, i.e.~when the dynamics are not themselves captured correctly. However, this is not the case for the ToastB234 (Figure S10). Therefore, we suggest that factor (1) above causes ME to fail in calculating the final size correctly, while factor (2) does not.

\subsection{Conjectures}
\label{sec:LargeConj}

Given the intuition developed in all the previous sections, for the Markovian SIR model with transmission rate
$\tau$ and recovery rate $\gamma$ as most commonly considered in the literature,
 we expect that errors in moment closure
schemes will be introduced by the following factors:
\begin{enumerate}
	\item Finite length of infectious period. Given $\tau/\gamma$ is the only
		dimensionless parameter in the model, we conjecture that these errors
		will be $O(\gamma / \tau)$.
	\item Long loops and some overlapping triangles. Where there is a clustering
		coefficient $\phi$ and triangles are introduced in a combinatorially 
		random manner, we conjecture such errors are $O(\phi^2)$.
\end{enumerate}
Of course, as the epidemic spreads, errors can accumulate, so we expect the
epidemic at larger times to be less accurate than at smaller times.

\section{Conclusions}
\label{sec:Concl}

We have presented here a detailed examination of the behaviour of
the most commonly used moment closure approximations, with particular
attention to the newly proposed approximation based on the concept of maximum
entropy. On an open triplet, this approximation coincides with the one
commonly used in the literature. On a closed triangle, instead, the ME
approximation is substantially more complex than the commonly used Kirkwood
approximation, but overcomes its fundamental theoretical drawbacks and,
overall, seems to perform better.

One of the interesting results from our work is that, when moving away from the commonly
considered Markovian assumption, the perspective can change dramatically, with all approximations being actually exact on the closed triangle when the infectious periods have
constant duration (Proposition \ref{thm:ClosedMEConst}). This agrees with the
intuition that we are trying to reconstruct a joint distribution through a
product of marginals, which is likely to work only when an assumption of
independence holds.

On larger networks, we have provided a simpler proof of the result of
\citet{Sharkey:2013} concerning the exactness of moment closure for the SIR
model on tree-like networks under pure initial conditions. Our proof also
extends the result to more general models. Concerning larger network with
clustering, the extensive numerical investigation we have performed on small motifs suggests ME
allows exact closure at the level of pairs on some large networks with
non-overlapping triangles for both SI and Reed-Frost-type models when a single
initial infective is present. Large scale numerical simulations confirm such conclusions for the SI model.

Moving on from exactness of moment closure to the quality of the approximations still requires substantial work. For example, even on the simple closed triangle, none of the closure techniques appears to
be uniformly better than any other, and the heterogeneity of their quality over
different transitions $\bx^0\to\bx$ suggests that the choice of which one
performs best will still be context-dependent. This was already noticed by
Rogers \cite{Rogers:2011}, by showing that in a specific example on an SIR
epidemic spreading on a small-world network, the ME approximation can still
lead to a worse overall performance than Kirkwood's.  Rogers claims this is due
to a fortunate error cancellation, where the underestimation in Kirkwood's
approximation of the number of susceptibles in closed triangles in the network
is compensated by its overestimation in open triplets. This appears
incorrect in light of Theorem \ref{thm:TreeAllModels} and the work of
\citet{Sharkey:2013}. For a small-world network, however, there is
non-negligible presence of short loops larger than a triangle and we believe
that the failure of \eqref{MEopen} for the open triplets that form a square is
likely to be the actual cause of the improved
performance of Kirkwood.

We hope the intuition built up through this extensive exploration can open many lines of thought from researchers in the epidemic modelling community and beyond. In particular we believe that it represents a valuable step in unravelling the assumptions behind local moment closure on networks. Without this understanding, there is arguably no hope to control the errors that build up in global moment closure approximation schemes. Given their versatility and the significant dimensionality reduction they can achieve, the ability to control their errors and to put them on a solid mathematical footing would represent a key and much desired methodological achievement.

\section*{Acknowledgements}
We gratefully acknowledge the Engineering and Physical Sciences Research Council for funding and the two anonymous reviewers for comments that lead to a substantially improved version of this manuscript.




\bibliographystyle{elsarticle-harv}



\newpage

\begin{figure}[H]
\centering
\vspace{2em}
\includegraphics[width=\textwidth]{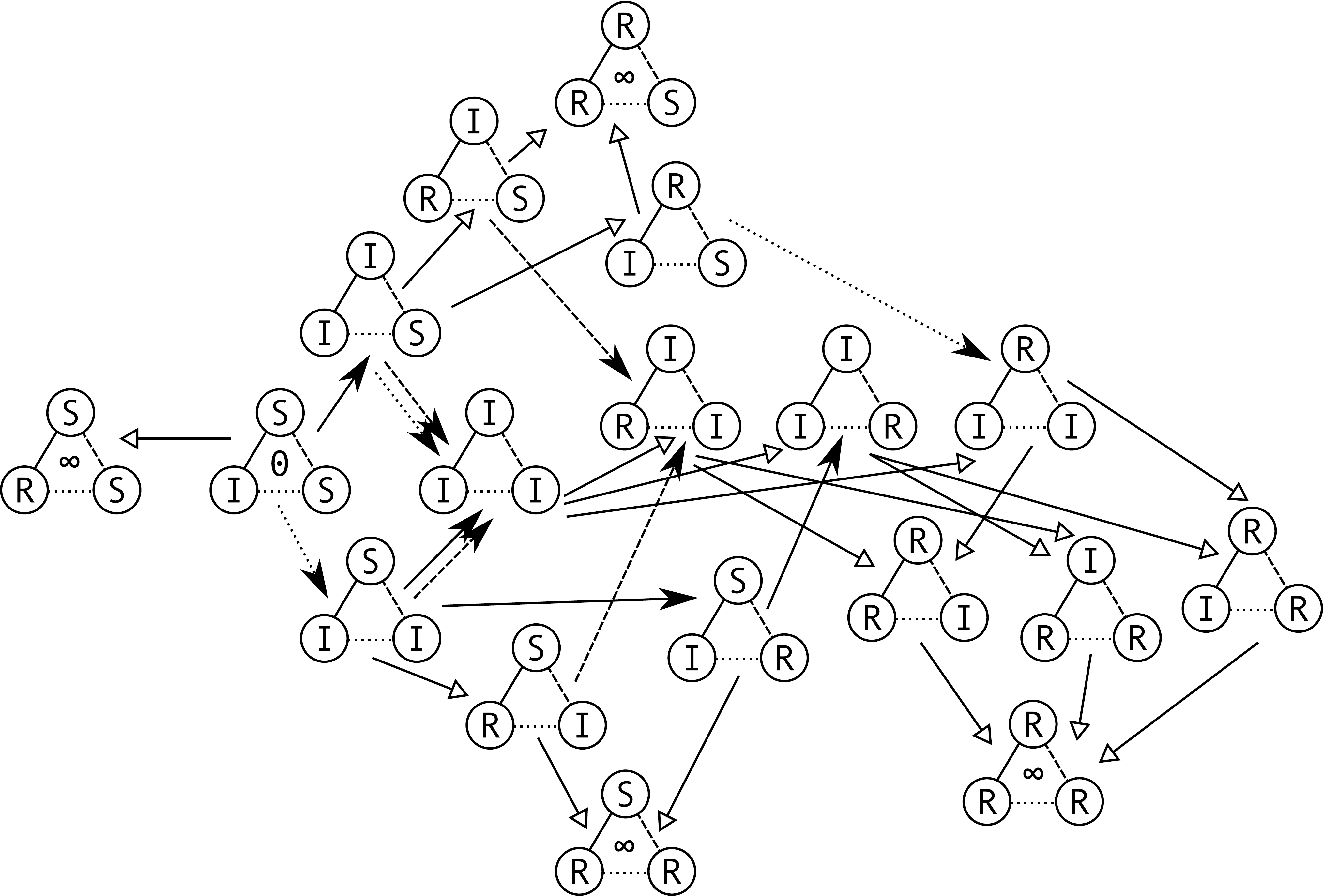}
\vspace{2em}
\caption{Set of states and transitions for the SIR model on the open triplet
and closed triangle. The starting point is marked $0$, and absorbing states
$\infty$. Open-headed arrows relate to recovery, and filled ones to transmission.
All lines are present for the closed triangle, and either the dotted or dashed
lines are absent depending on the initial conditions of the open triple.}
\label{fig:sirmodel} 
\end{figure}

\newpage

\begin{figure}[H]
\centering
\includegraphics[width=\textwidth]{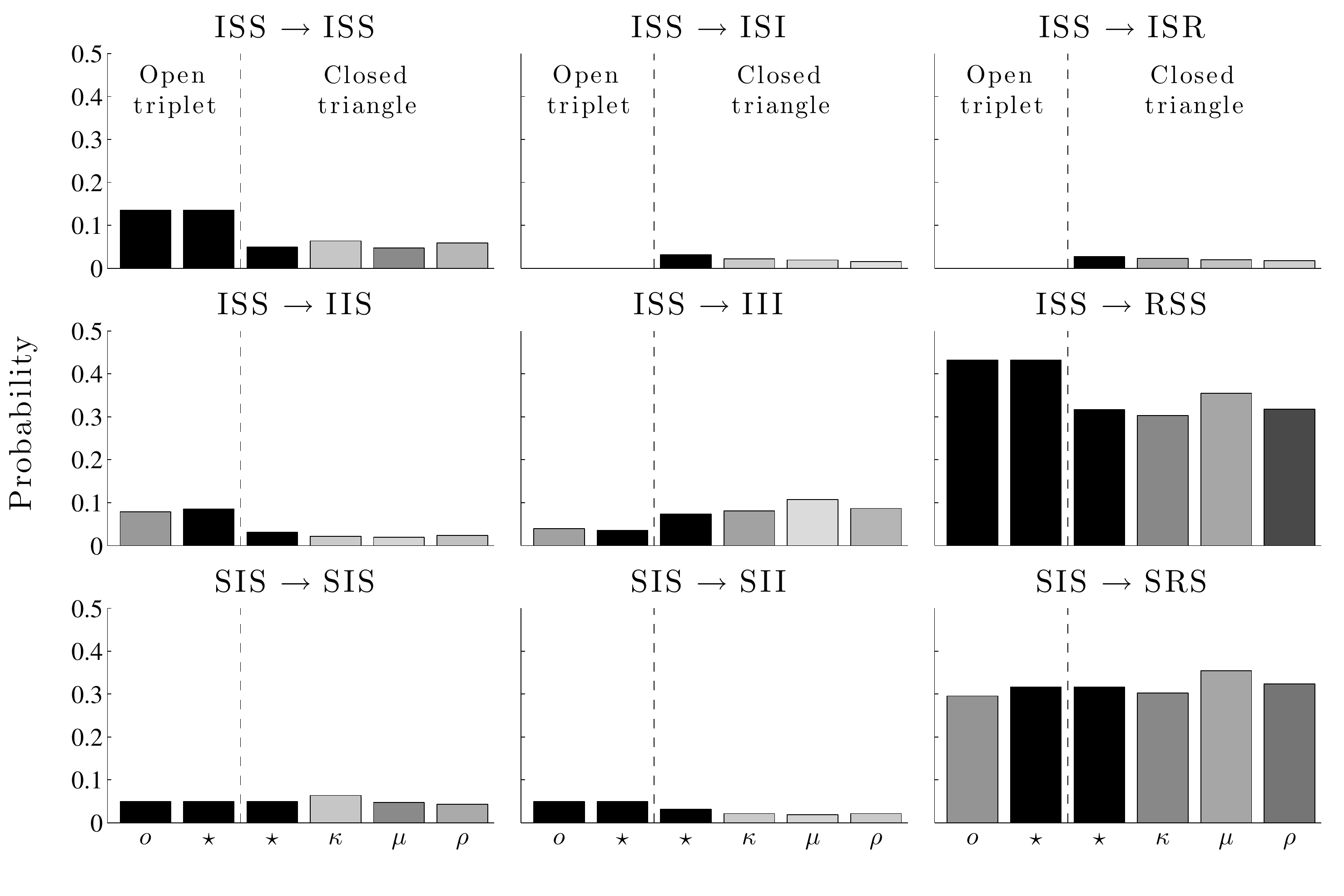}
\caption{Exact ($\star$) and approximate ($o, \mu, \kappa$ and $\rho$) probabilities for an open
triplet (left two bars) and a closed triangle (right four bars) being in state
$\bx$ at time $t=1$ when starting from state $\bx^0$ at time $t=0$, for some
selected cases $\bx^0\to\bx$. For the open triplet and the closed triangle,
respectively, the exact probability is coloured in black while the lighter the
shade of grey of each approximation, the larger its relative difference with
the (appropriate) exact probability (black = 0\%, white = 100\%). The Markovian
model is assumed, with infectivity $\tau = 1$ and average duration of the
infectious period $\meanT=1$.
}\label{fig:MarkExamples}
\end{figure}

\newpage

\begin{figure}[H]
\centering
\includegraphics[width=\textwidth]{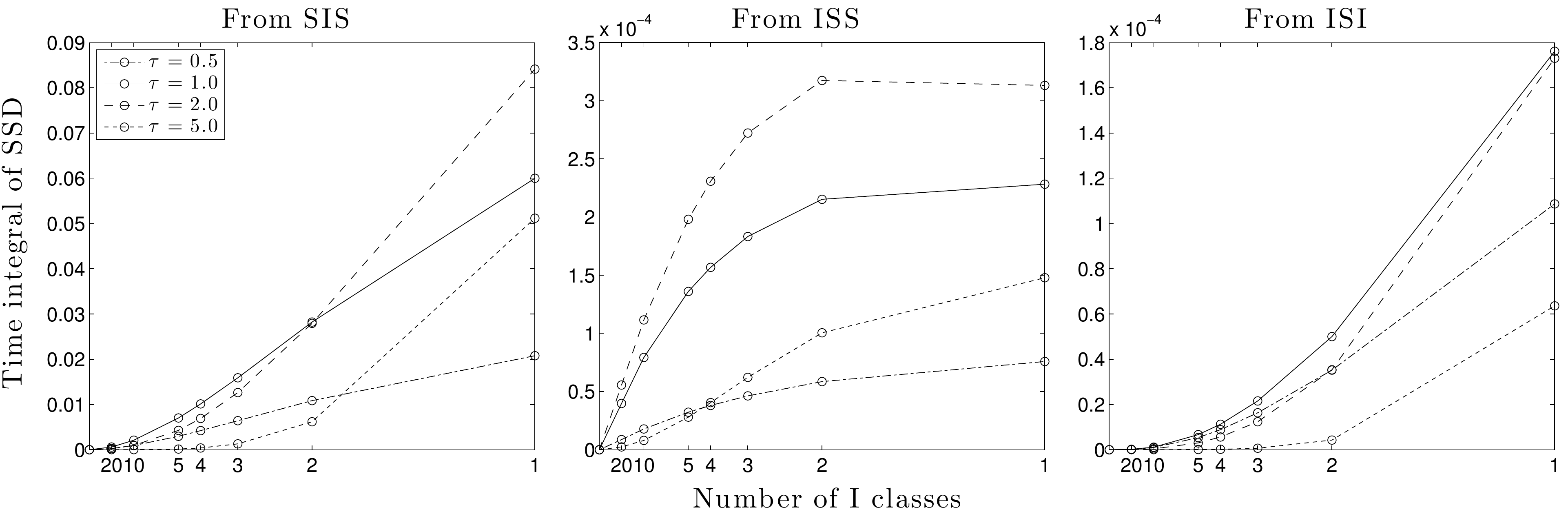}
\caption{Time integral of the sum of squared difference (SSD) between the exact
and the approximate probability distributions over all states $\bx$ of an open
triplet, starting from each of the three states $\bx^0=(SIS), (ISS)$ and $(ISI)$, as a function of the number of infectious classes in the SIR model, for various values of the infectivity $\tau$. The $x$-axis is scaled so that the variance of the duration of the infectious period in the presence of $n_I$ (equally infectious) classes, $\Var{T}=1/n_I$, appears increasing linearly.}\label{fig:OpenSSDVSvar}
\end{figure}

\begin{figure}[H]
\centering
\includegraphics[width=\textwidth]{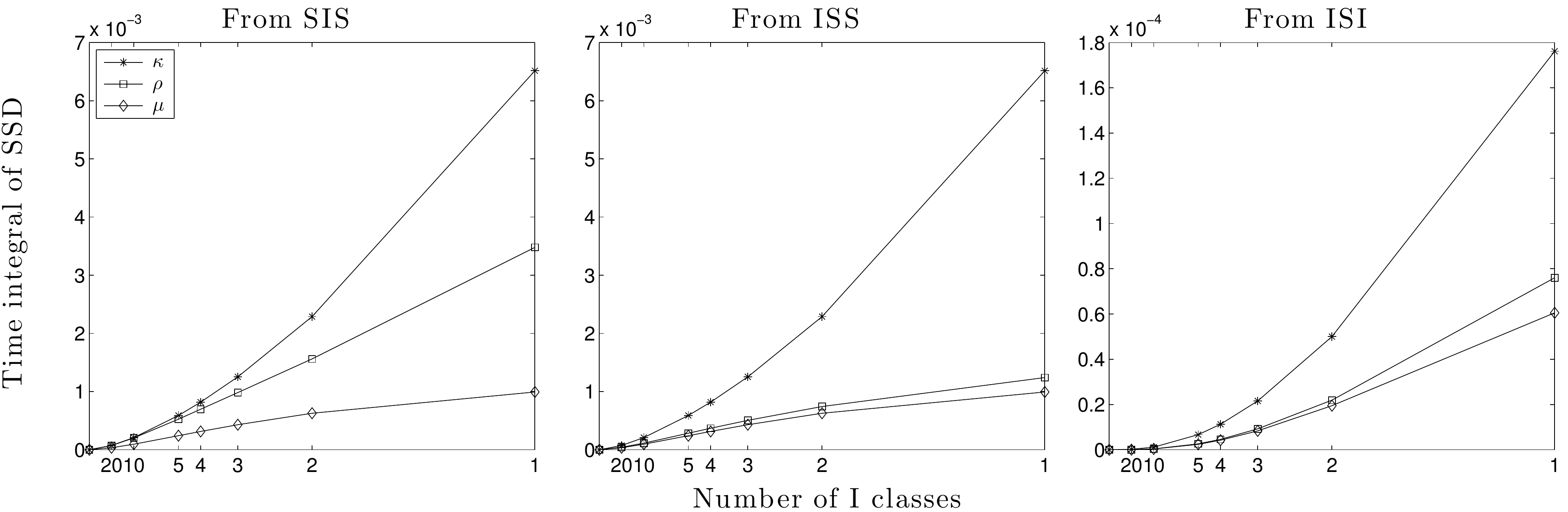}
\caption{Time integral of the sum of squared difference (SSD) between the exact
probability distributions over all states $\bx$ of a closed triangle and each
of the three approximations,  starting from each of the three states
$\bx^0=(SIS), (ISS)$ and $(ISI)$, as a function of the number of infectious classes of the SIR model ($x$-axis linearly increasing with the variance).
}
\label{fig:ClosedSSDVSvar} 
\end{figure}

\newpage

\begin{figure}[H] 
\centering
\includegraphics[width=\textwidth]{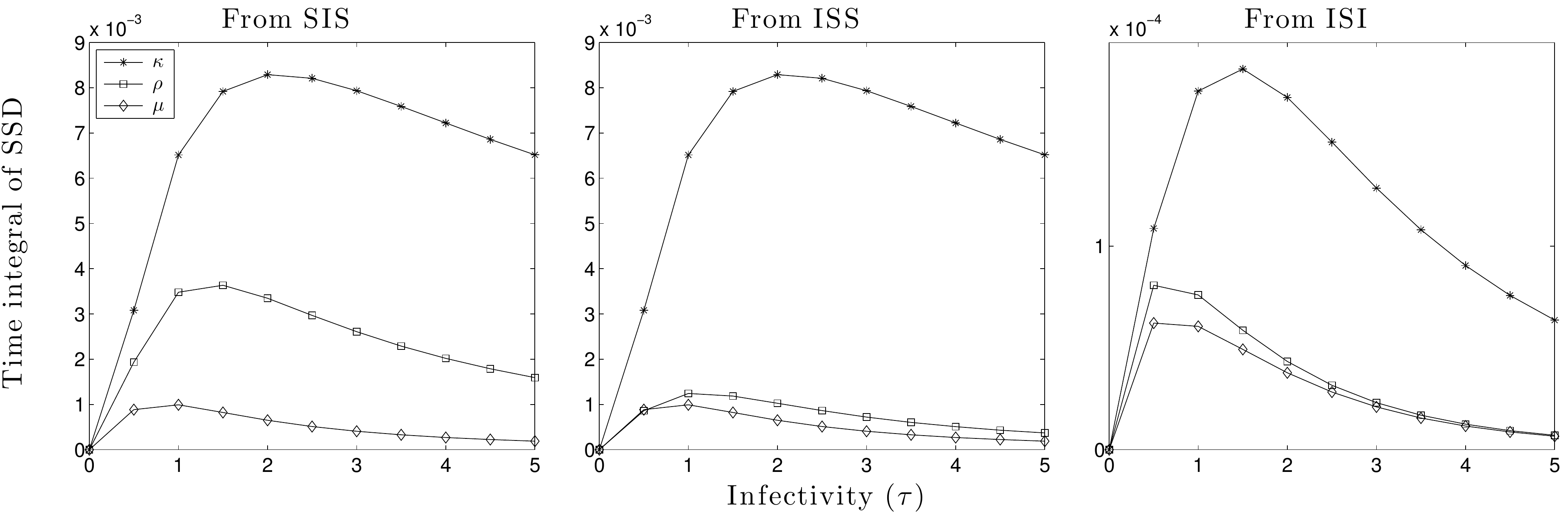}
\caption{Time integral of the sum of squared difference (SSD) between the exact
probability distributions over all states $\bx$ of a closed triangle and each
of the three approximations,  starting from each of the three states
$\bx^0=(SIS), (ISS)$ and $(ISI)$, as a function of the infectivity $\tau$ of the Markovian SIR model.
}\label{fig:ClosedSSDVStau}
\end{figure}

\begin{figure}[p]
\centering
\includegraphics[scale=0.4]{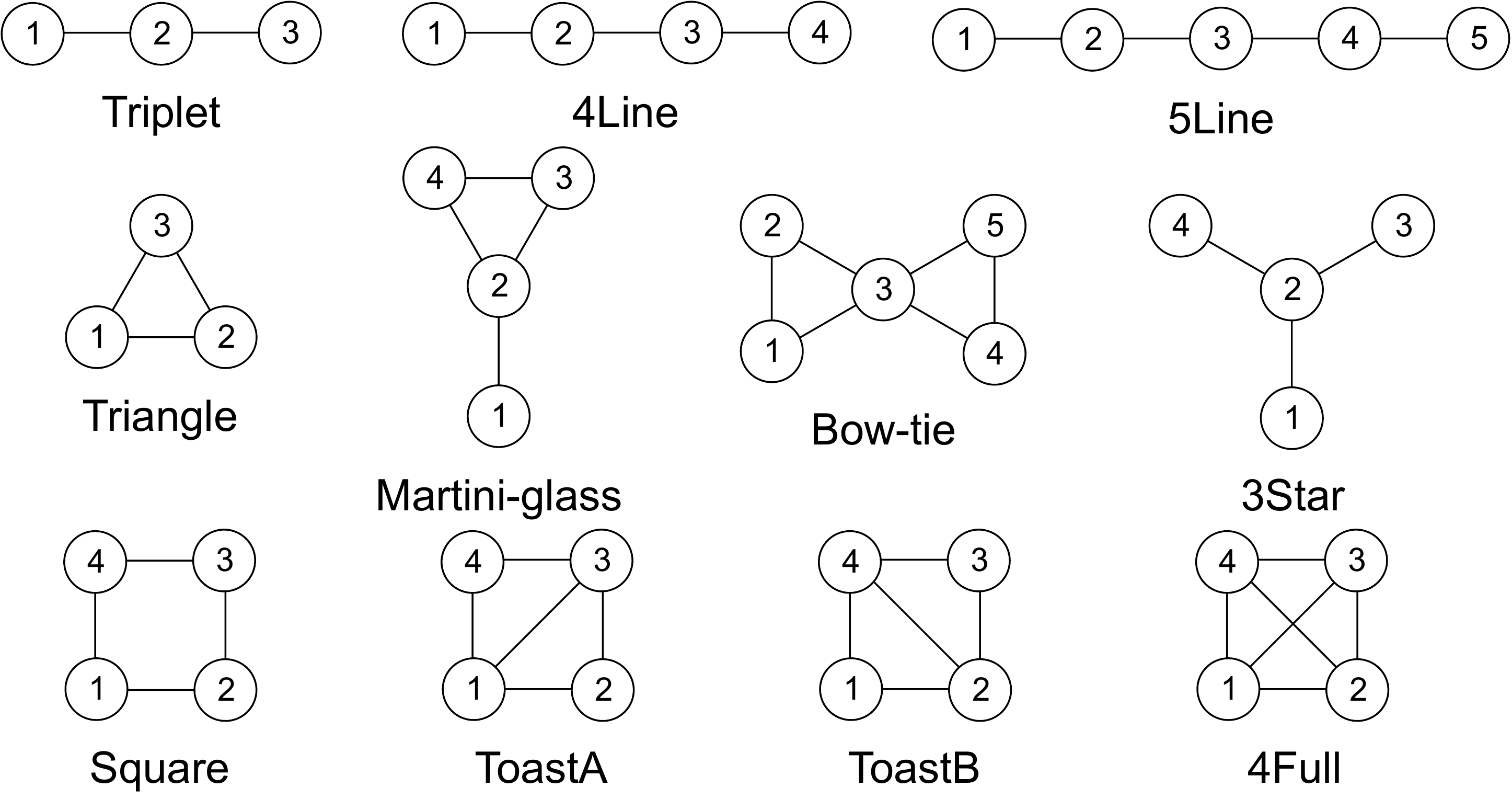}\\
\includegraphics[scale=0.4]{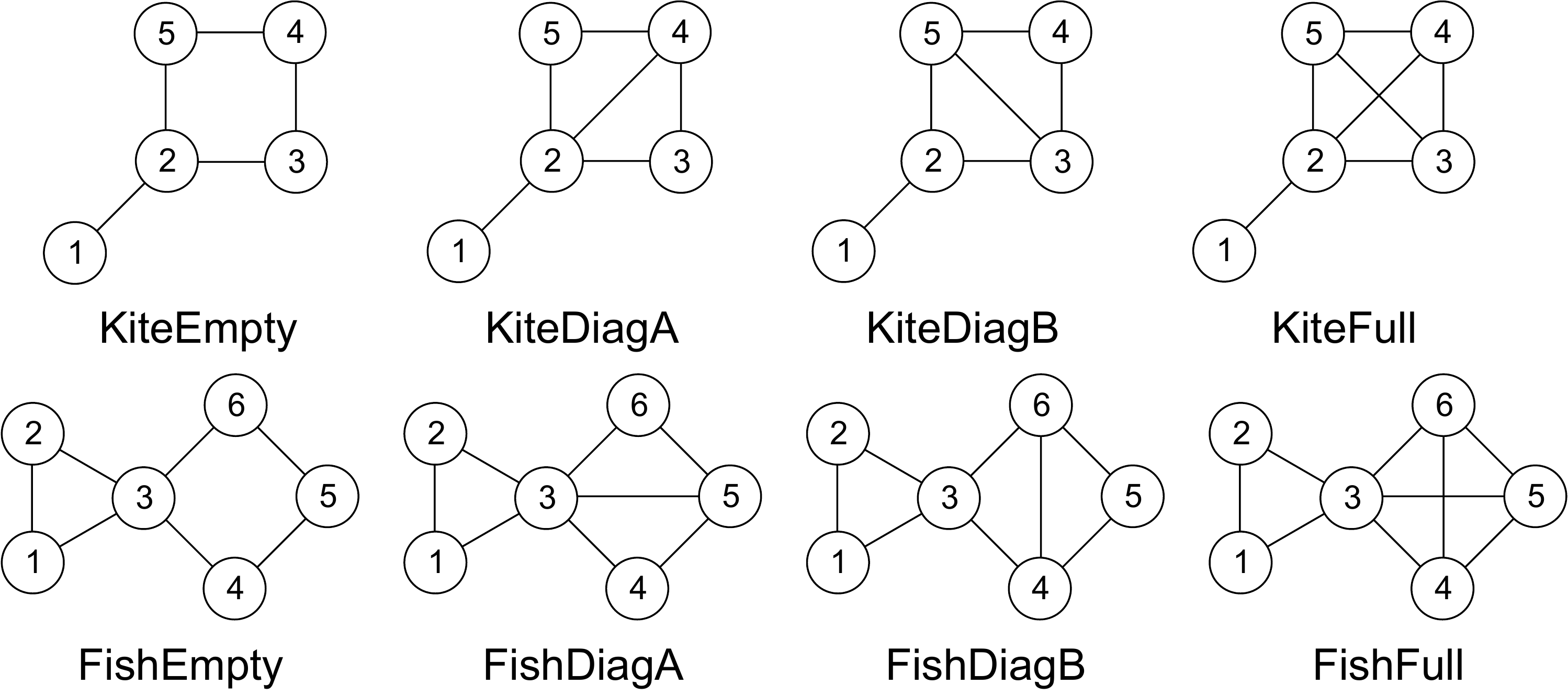}\\
\includegraphics[scale=0.4]{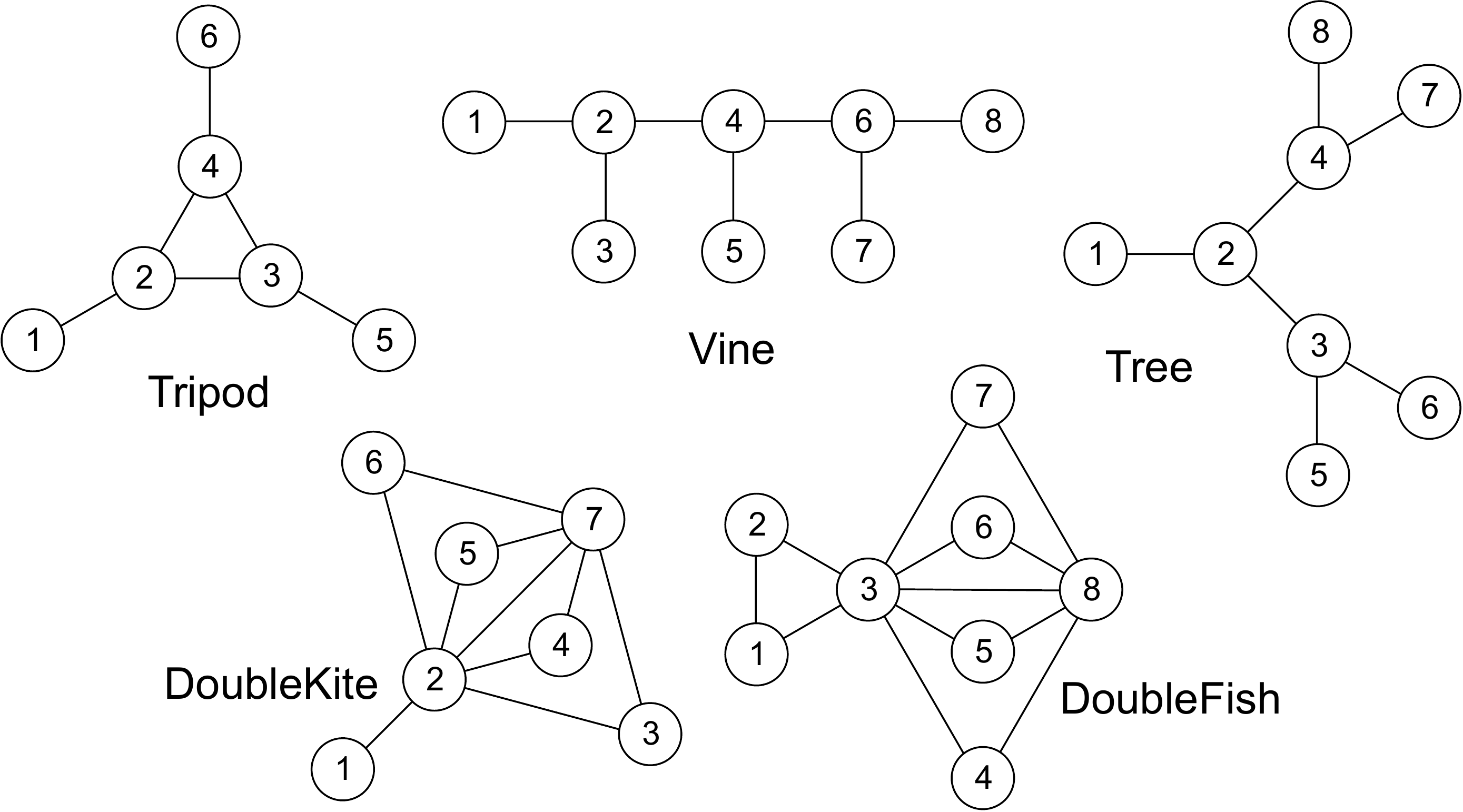}
\caption{Motifs analysed in Section \ref{sec:OtherMotifs}.}
\label{fig:graphs}
\end{figure}

\newpage

\begin{table}[p]
\centering
\includegraphics[height = 12.5cm]{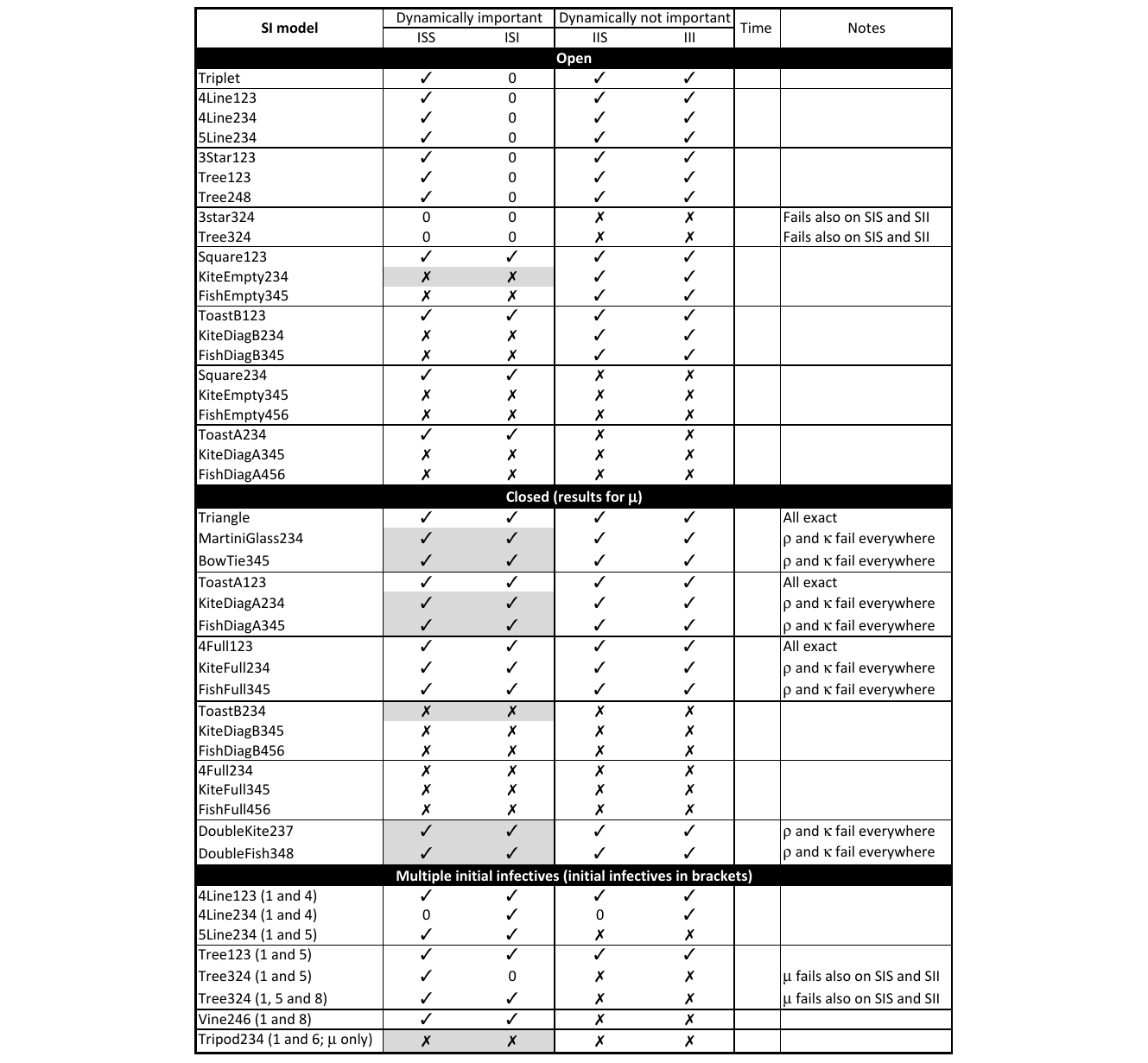}
\caption{Exactedness of moment closures at the level of triplets for the SI model. Network names refer to Figure \ref{fig:graphs} and are appended with the list of nodes the closure is applied to. If not specified, the approximation is tested at $t=0.5$. Closures can: ``work'', i.e.~be exact (\ding{51}) at the time tested, suggesting general validity; ``fail'' to be exact (\ding{55}); or refer to a state that is never reached by the system (0), in which case they work but provide no useful information about their general validity. Grey areas highlight test results that provide key understanding and that are discussed in the main text.}
\label{tab:SI}
\end{table}

\newpage

\begin{table}[p]
\centering
\includegraphics[height = 12.5cm]{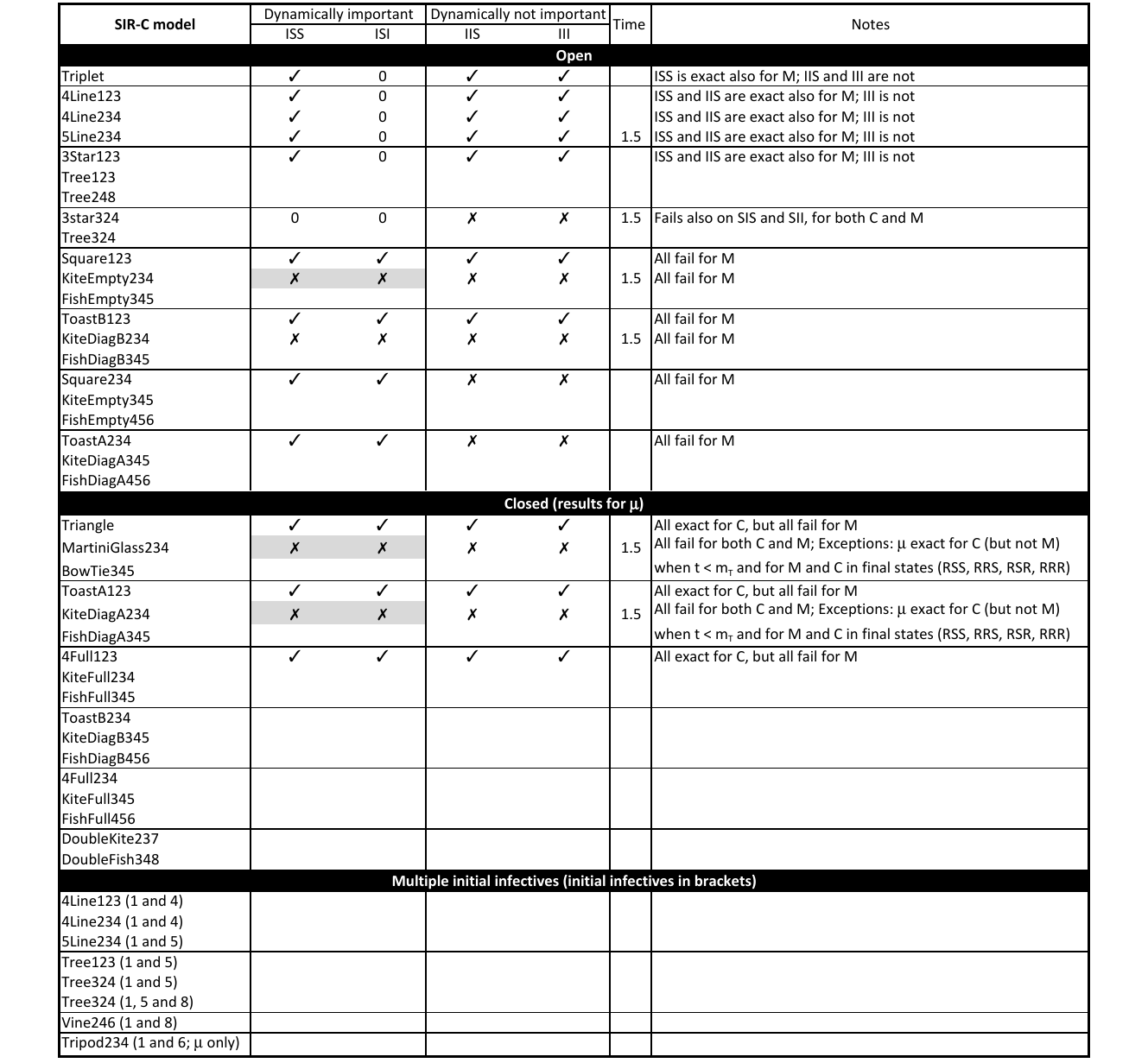}
\caption{Exactedness of moment closures at the level of triplets for the SIR model with a constant duration of the infectious period (C). Comments for the Markovian model (M) with exponentially distributed duration of infection are also reported when useful. Time of test is $t=0.5$ when not stated. Only the interesting results are reported. Symbols and table structure are as per Table \ref{tab:SI}.}
\label{tab:SIR}
\end{table}

\newpage

\begin{figure}[H]
\centering
\includegraphics[width = \textwidth]{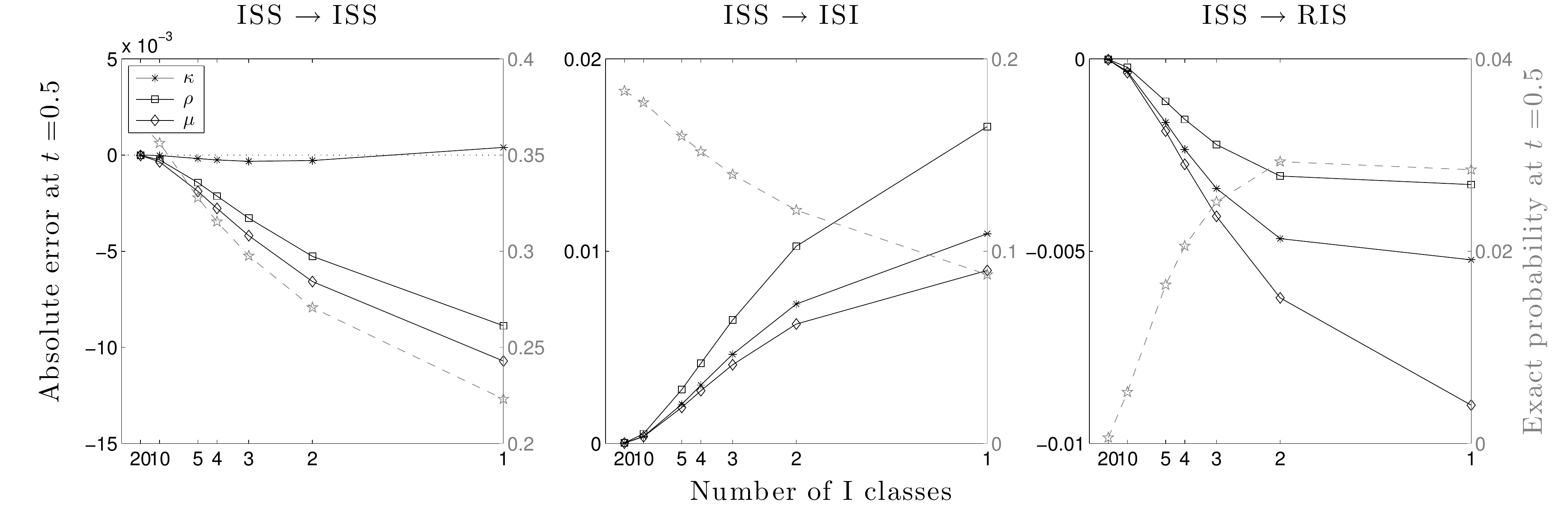}
\caption{Investigation of the error in the three moment closures (Kirkwood, $\kappa$; 1-step ME, $\rho$; and ME, $\mu$) for the SIR model on a closed triangle at time $t=0.5$. The left axes (grey dashed lines with 5-point star markers) shows the probability that at $t=0.5$ the system is in the state of interest. In addition to the two dynamically important states $ISS$ and $ISI$, we plotted the results for $RIS$, as an example of a state that, in the SIR-C model occurs with negligible probability at time $t=0.5$.}
\label{fig:TriangleSIR}
\end{figure}

\begin{figure}[H]
\centering
\includegraphics[width = \textwidth]{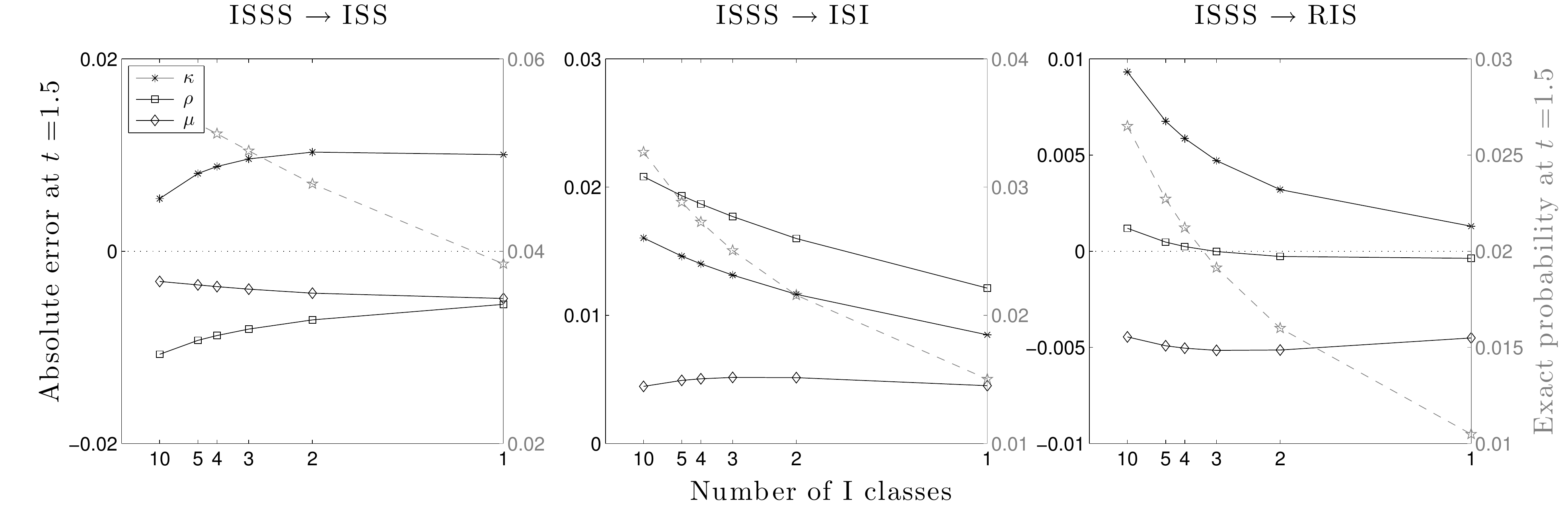}
\caption{Investigation of the error in the three moment closures (Kirkwood, $\kappa$; 1-step ME, $\rho$; and ME, $\mu$) for the SIR model on the MartiniGlass234 network at time $t=1.5$. The left axes (grey dashed lines with 5-point star markers) shows the probability that at $t=1.5$ the system is in the state of interest. In addition to the two dynamically important states $ISS$ and $ISI$, we plotted the results for $RIS$. All states occur with positive probability at $t=1.5$, for all models from M to C, and no closure is exact.}
\label{fig:MartiniGlass234_SIR}
\end{figure}

\newpage

\begin{table}[p]
\centering
\includegraphics[height = 12.5cm]{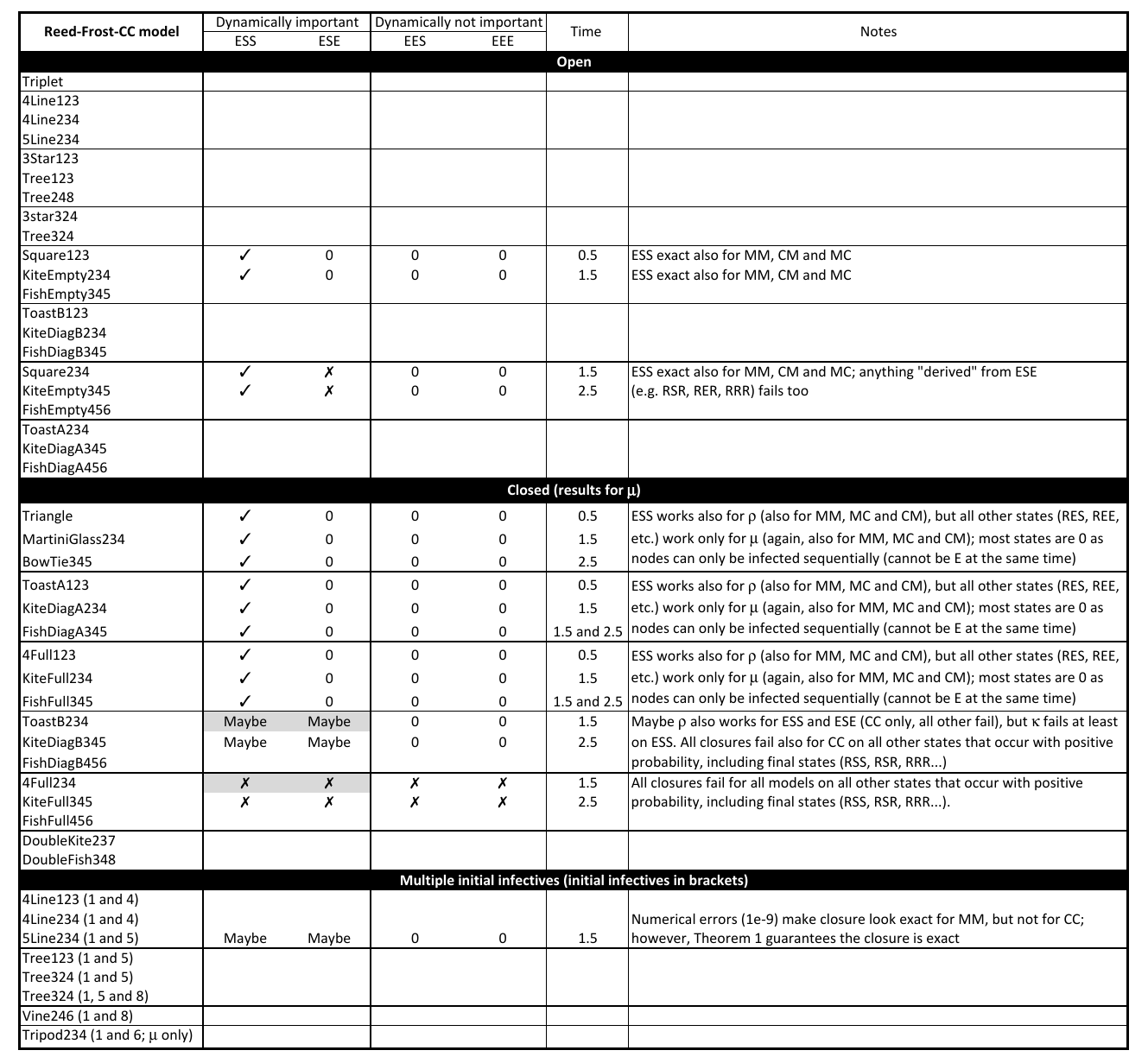}
\caption{Exactedness of moment closures at the level of triplets for the Reed-Frost model with a constant duration of the latent period and non-random probabiliry $P\equiv p$ of transmission (CC). Comments for exponentially distributed latent period or geometrically distributed probability of tranmission $P$, or both (MC, CM or MM, respectively) are also reported when useful. Time of test is $t=0.5$ when not stated. Only the interesting results are reported. Symbols and table structure are as per Table \ref{tab:SI}.}
\label{tab:RF}
\end{table}

\newpage

\begin{figure}[H]
\centering
\includegraphics[width = \textwidth]{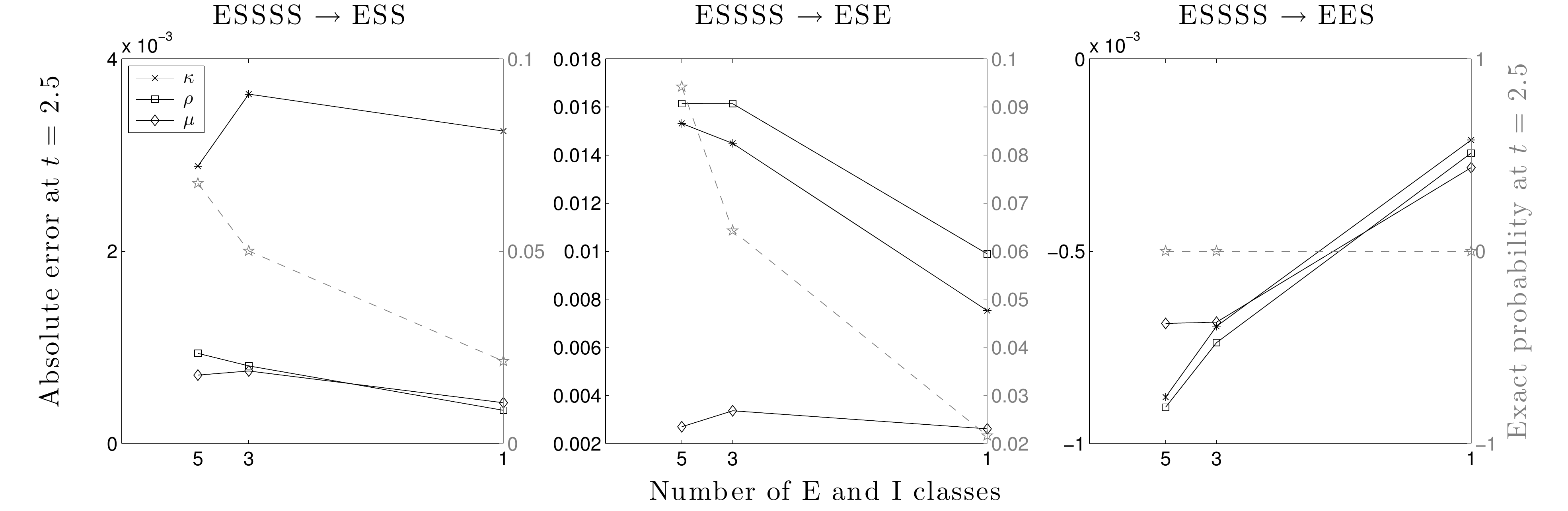}
\caption{Investigation of the error in the three moment closures (Kirkwood, $\kappa$; 1-step ME, $\rho$; and ME, $\mu$) for the RF model on the KiteDiagB345 network at time $t=2.5$. The left axes (grey dashed lines with 5-point star markers) shows the probability that at $t=2.5$ the system is in the state of interest. Note that state EES never occur with positive probability, because individuals 3 and 4 can never be infected at the same time: if 4 is in the $E$ state, 3 was either a potential infector (and so is now in the $R$ state) or has escaped the infection from 2 and is therefore in state $S$. Also note the dubious convergence, that is difficult to investigate because of the computational cost involved.}
\label{fig:KiteDiagB345}
\end{figure}

\begin{figure}[H]
\centering
\includegraphics[width = \textwidth]{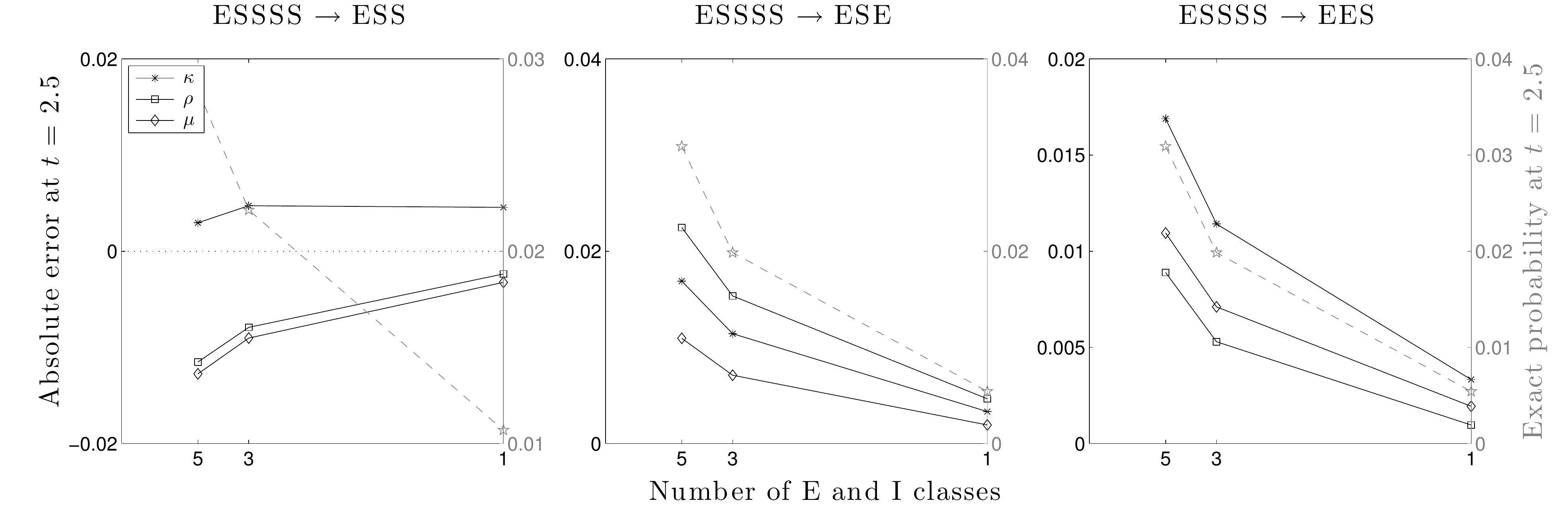}
\caption{Investigation of the error in the three moment closures (Kirkwood, $\kappa$; 1-step ME, $\rho$; and ME, $\mu$) for the RF model on the KiteFull345 network at time $t=2.5$. The left axes (grey dashed lines with 5-point star markers) shows the probability that at $t=2.5$ the system is in the state of interest. Note now the clear lack of convergence.}
\label{fig:KiteFull345_RF}
\end{figure}

\newpage

\begin{figure}[p]
\centering
\subfloat[]{ \scalebox{0.4}{\resizebox{\textwidth}{!}{%
\includegraphics{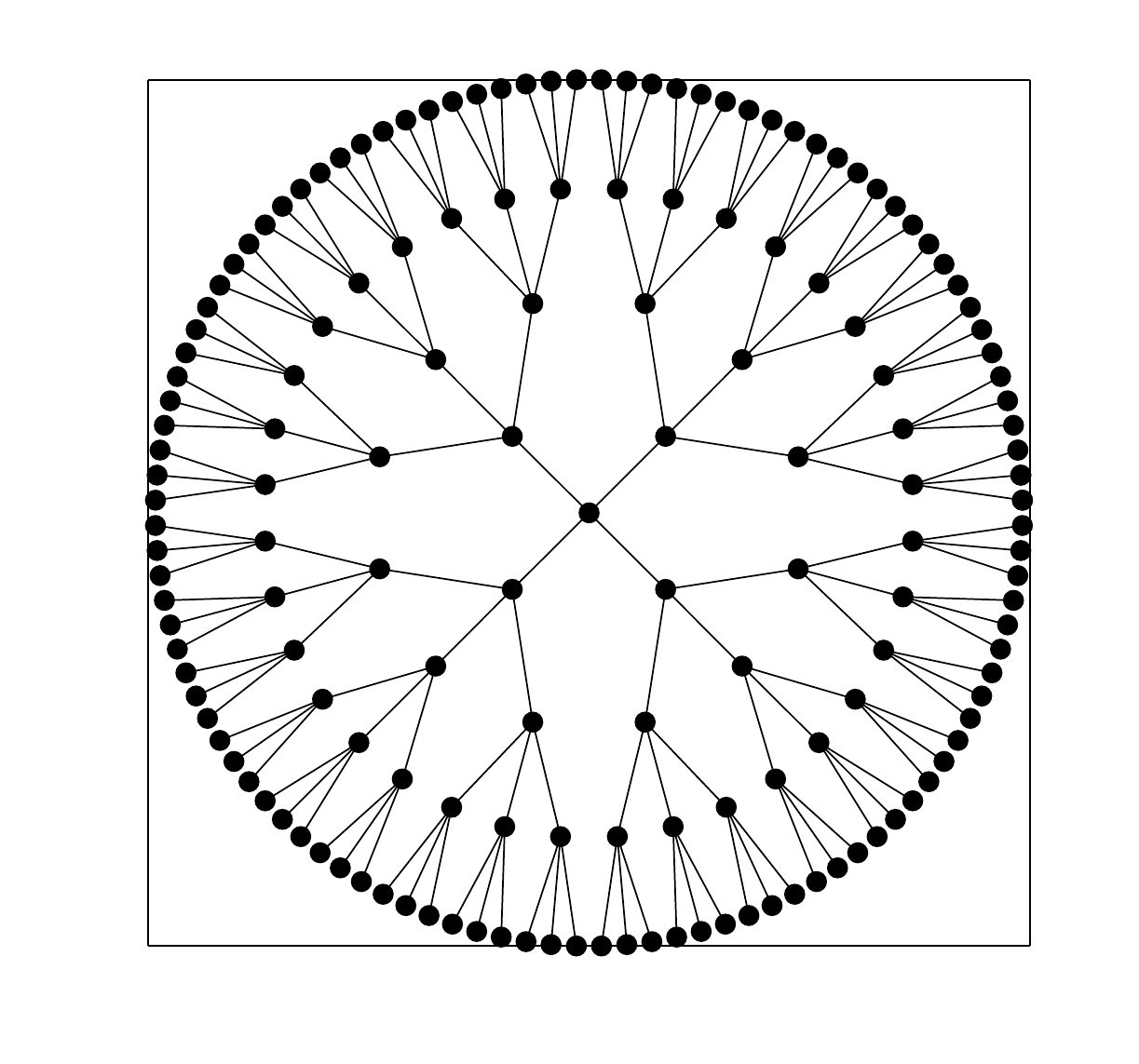} }}} \quad
\subfloat[]{ \scalebox{0.5}{\resizebox{\textwidth}{!}{%
\includegraphics{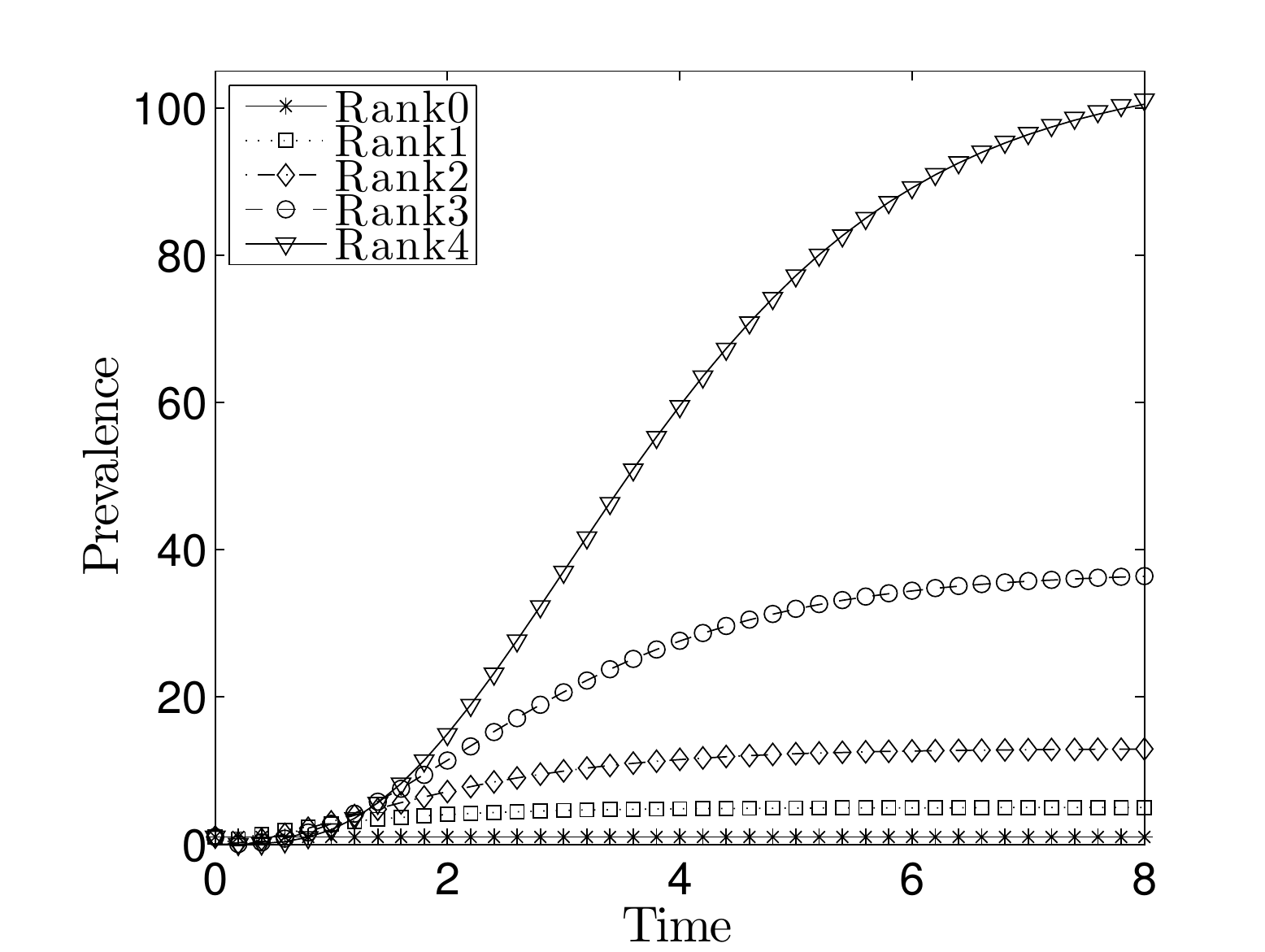} }}}\\
\subfloat[]{ \scalebox{0.4}{\resizebox{\textwidth}{!}{%
\includegraphics{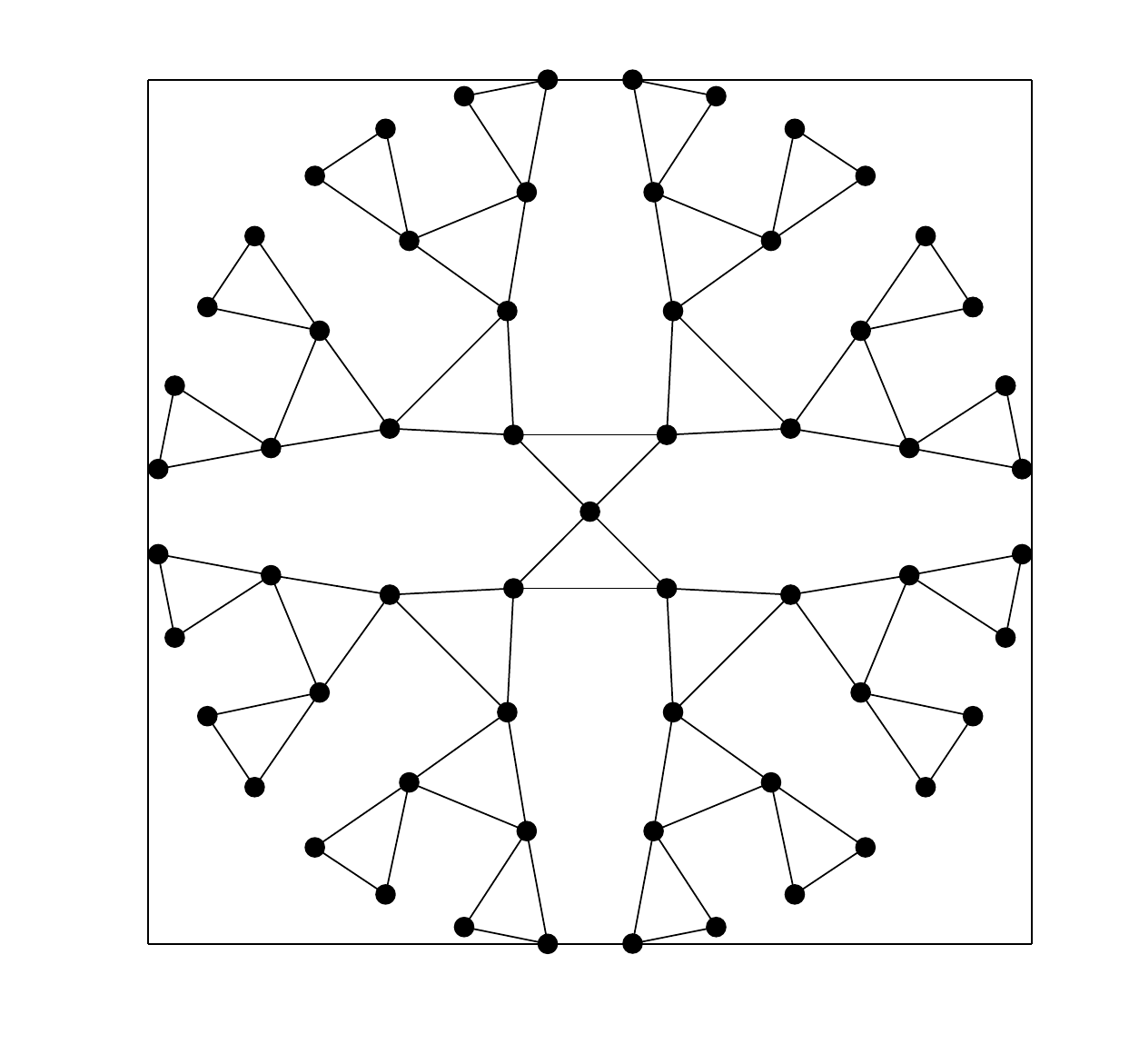} }}} \quad
\subfloat[]{ \scalebox{0.5}{\resizebox{\textwidth}{!}{%
\includegraphics{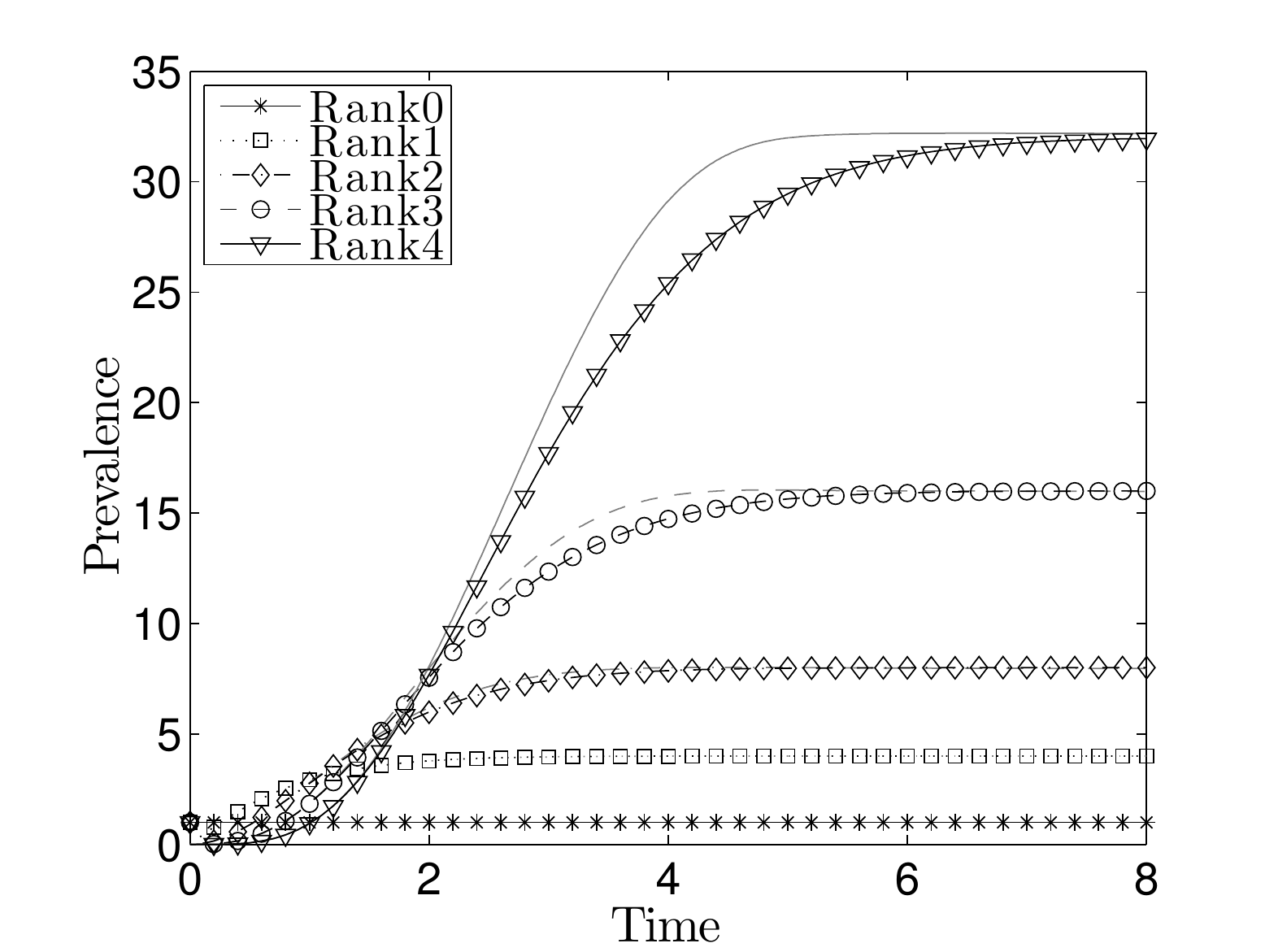} }}}
\caption{SI dynamics on larger networks. (a) A tree network. (b) Mean numbers
	infective over time at different Rank (distance from the central node) for
	the tree network, for Monte Carlo simulation (markers) and exact ODE models
	(lines). (c) A tree-of-triangles network. (d) Mean numbers infective over
time at different Rank for the tree-of-triangles network, for Monte Carlo
simulation (markers), inexact Kirkwood ODEs (grey lines), and exact Maximum
Entropy ODEs (black lines).  \label{fig:sitrees}}
\end{figure}

\newpage

\begin{center}
\Large \textbf{Exact and approximate moment closures for non-Markovian network
epidemics} \\ \vspace{4mm}
\Large \textit{Supplementary Material} \\ \vspace{4mm}
\large  Lorenzo Pellis $\qquad$ Thomas House $\qquad$ Matthew J.\ Keeling
\end{center}

\section*{Open triplet}\label{sec:OpenSupp} 
Figure 2 of the main text hides in an aggregate measure most of the
heterogeneity in the performance of the standard approximation for an open
triplet. Figure \ref{fig:Open12plots} unravels some of this heterogeneity,
revealing positive and negative errors in different cases, exactness in others
and a particularly poor performance when starting from state $\bx^0=(SIS)$.
Figure \ref{fig:Open12plots_VS_tau} plots the same results as a function of the
infection rate $\tau$.


\section*{Closed triangle}\label{sec:ClosedSupp} 
The complex behaviour of the three approximations in all different cases makes
it difficult to have a full overview of their accuracy. Here we finally present
an almost exhaustive list of all interesting $\bx^0\to\bx$ cases. Figures
\ref{fig:ErrorVStime1-9} and \ref{fig:ErrorVStime10-18} shows the error
$e_c^{\bx^0}(\bx;t)$ as a function of time for the Markovian model with
infectivity $\tau=1$
. Note, as already observed in the main text, how ME performs poorly compared
to Kirkwood's for the case $(ISS)\to(ISS)$, how Kirkwood's approximation is
strongly inaccurate for $(ISS)\to(III)$ and how both fail to capture correctly
the case $(ISS)\to(RSS)$ (although the relative performance of ME improves
dramatically for larger values of $\tau$; not shown). Note also how both ME and
Kirkwood's approximations give the same results for the cases $(ISS)\to(ISR)$
and $(ISS)\to(IRS)$ as they are symmetrical, but 1-step ME does not. Further
exploration of how 1-step ME performs when reaching the same state $\bx$ from
all three initial states $\bx^0=(ISS), (SIS)$ and $(SSI)$ reveals always an
identical behaviour in two out of the three cases, and a different behaviour
for the third one. Errors obtained when starting from $\bx^0=(ISI)$ are significantly smaller than
when starting from a single initial infective (Figure
\ref{fig:ErrorVStimeISI}). More strikingly, even though Kirkwood's
approximation appears to be quite inaccurate in general, it turns out to be
exact when in the particular cases of $\bx=(ISI), (ISR)$ (and thus $(RSI)$) and
$(RSR)$, when starting from $\bx^0 = (ISI)$.  

The heterogeneous behaviour
highlighted by Figure 5 in the main text suggests that ME, though better than
the other approximations in general, is not uniformly so. Figure
\ref{fig:iad_VS_tau_bestc_L1_diffy} explores how the time integral of the
absolute error $| e_c^{\bx^0}(\bx;t) |$ depends on $\tau$ in various cases of
interest. Note, first of all, how all errors converge to 0 for large $\tau$.
Second, note how ME is markedly inaccurate in the case $(ISS)\to(ISS)$, how
Kirkwood's performs poorly for $(ISS)\to(RSS)$ and $(RRR)$ while it is exact
for  $(ISI)\to(ISI)$, and how ME performs badly compared to Kirkwood's for small $\tau$ in the case $(ISS)\to(RRS)$.
Finally we report
in Figure \ref{fig:ClosedStackedSSD_9plots_diffy} the stratified contribution
to the overall SSD measure for each approximation and each starting point
$\bx^0$ of interest. As highlighted in the main text, in addition to showing a
quantitatively smaller discrepancy, ME seems to be always balancing the
discrepancy between exact results and approximations more evenly across states
and in time.

\section*{Supplementary Figures}

\begin{figure}[H] 
\centering
\includegraphics[width=\textwidth]{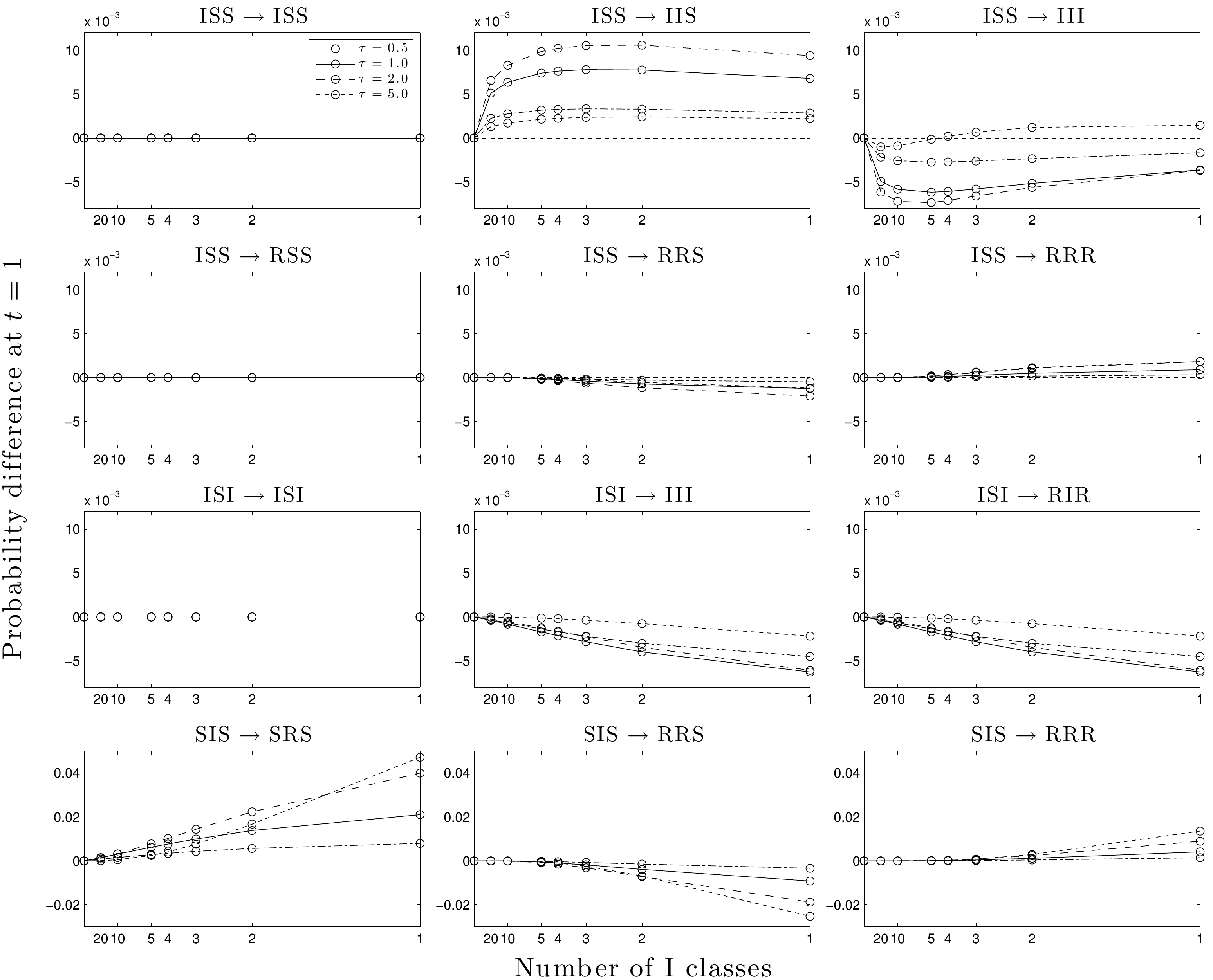}
\caption{Error $e_o^{\bx^0}(\bx;t)=\pb^{\bx^0}(\bx;t) - \pb_o^{\bx^0}(\bx;t)$ between the exact and approximate probabilities, in the SIR model, of an open triplet being in state $\bx$ at time $t=1$, when starting from state $\bx^0$ at time $t=0$, for various choices of $\bx^0$ and $\bx$, as a function of the number of infectious classes ($x$-axis linearly increasing with the variance), 
for various values of the infectivity $\tau$. Note the different scale of the $y$ axis of the bottom row. }\label{fig:Open12plots}
\end{figure}

\newpage

\begin{figure}[H] 
\centering
\includegraphics[width=\textwidth]{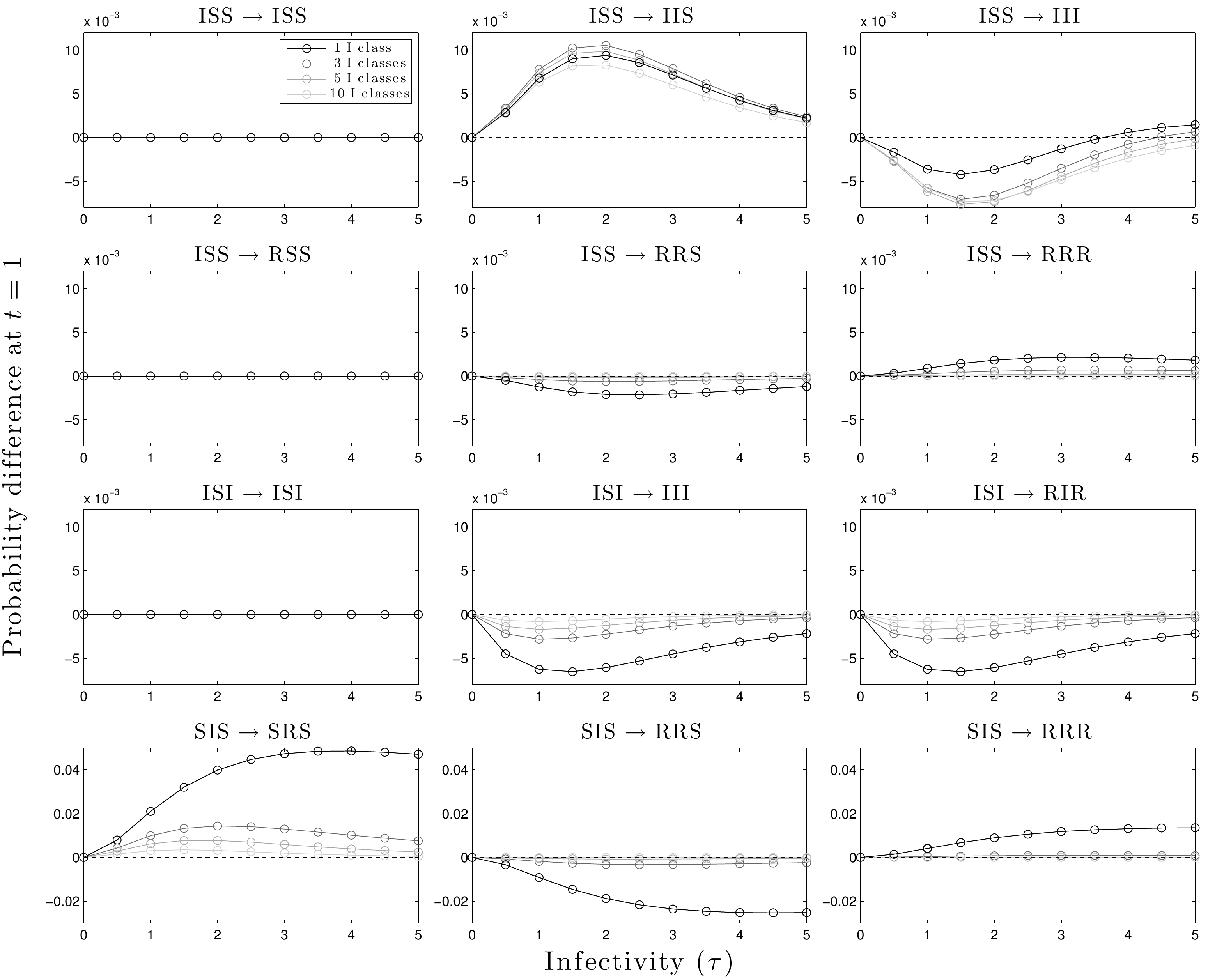}
\caption{Error $e_o^{\bx^0}(\bx;t)=\pb^{\bx^0}(\bx;t) - \pb_o^{\bx^0}(\bx;t)$ between the exact and approximate probabilities, in the SIR model, of an open triplet being in state $\bx$ at time $t=1$, when starting from state $\bx^0$ at time $t=0$, for various choices of $\bx^0$ and $\bx$, as a function of the infectivity $\tau$, for various number of infectious classes
. Note the different scale for the $y$ axis on the bottom row.}\label{fig:Open12plots_VS_tau}
\end{figure}

\newpage

\begin{figure}[H] 
\centering
\includegraphics[width=\textwidth]{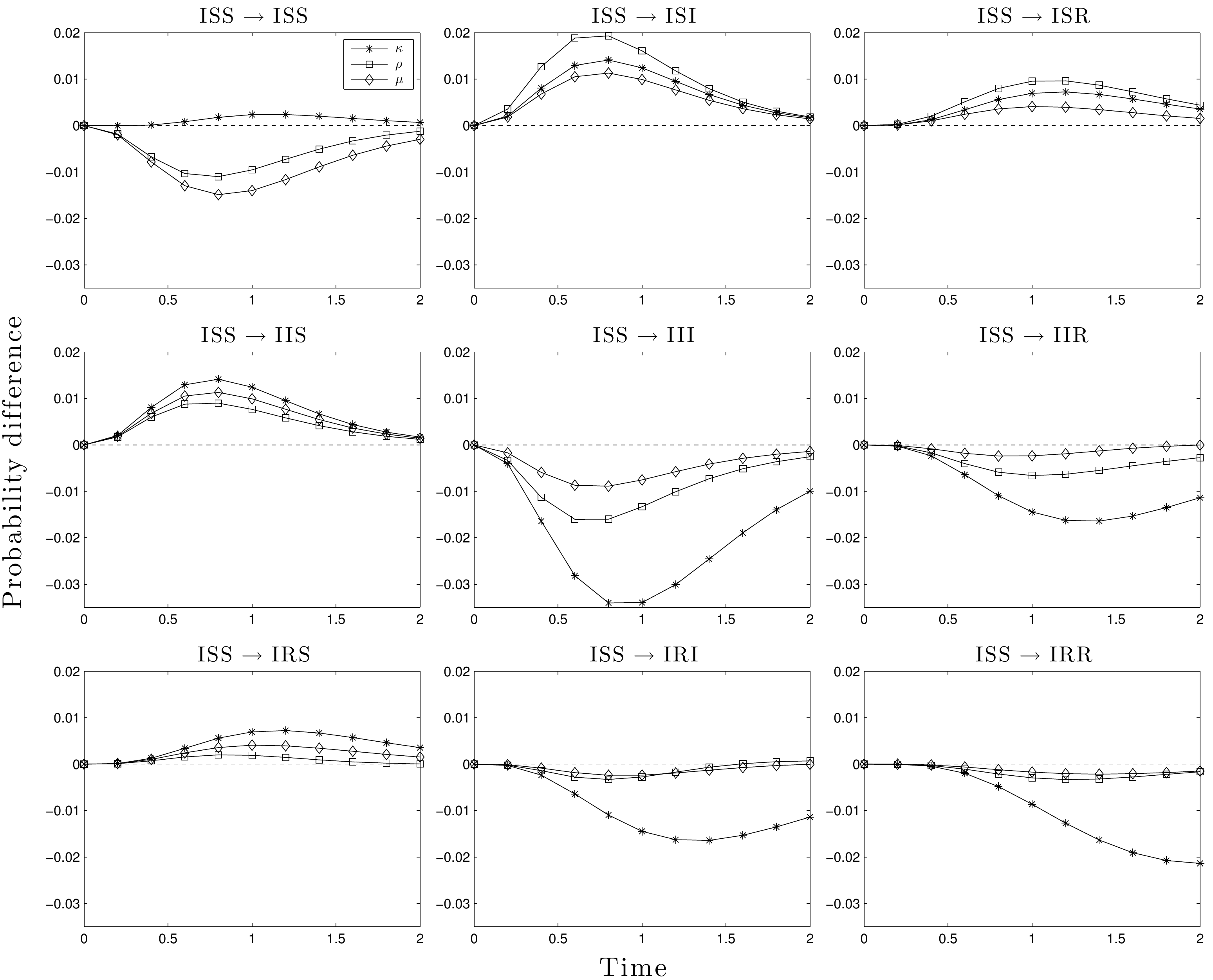}
\caption{Error $e_c^{\bx^0}(\bx;t)=\pb^{\bx^0}(\bx;t) - \pb_c^{\bx^0}(\bx;t)$ between the exact and approximate probabilities of a closed triangle being in state $\bx$ as a function of time, when starting from state $\bx^0$ at time $t=0$, for various choices of $\bx^0$ and $\bx$. The model is Markovian SIR with infectivity $\tau=1$.}\label{fig:ErrorVStime1-9}
\end{figure}

\newpage

\begin{figure}[H]
\centering
\includegraphics[width=\textwidth]{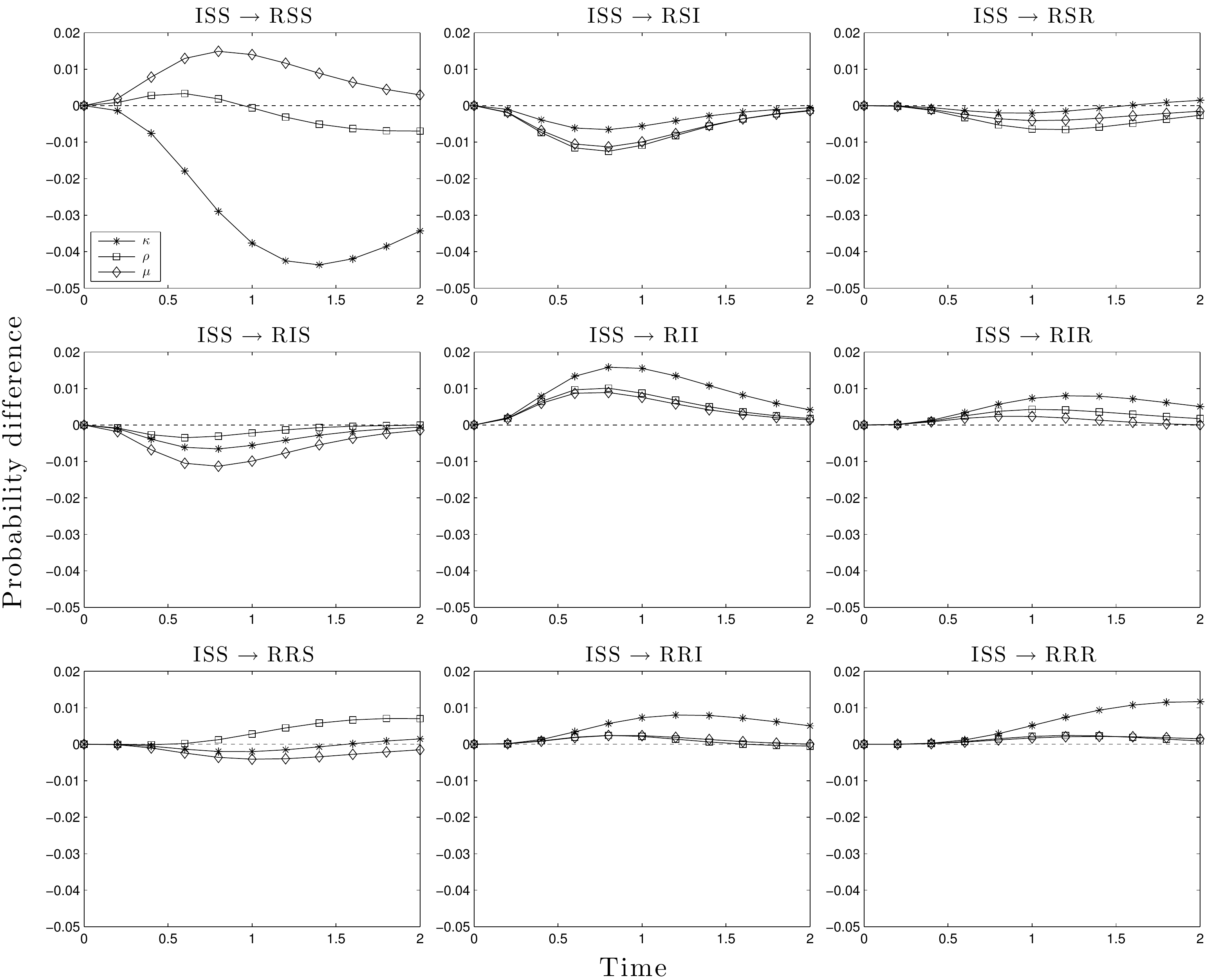}
\caption{Error $e_c^{\bx^0}(\bx;t)=\pb^{\bx^0}(\bx;t) - \pb_c^{\bx^0}(\bx;t)$ between the exact and approximate probabilities of a closed triangle being in state $\bx$ as a function of time, when starting from state $\bx^0$ at time $t=0$, for various choices of $\bx^0$ and $\bx$. The model is Markovian SIR with infectivity $\tau=1$.}\label{fig:ErrorVStime10-18} 
\end{figure}

\newpage

\begin{figure}[H] 
\centering
\includegraphics[width=\textwidth]{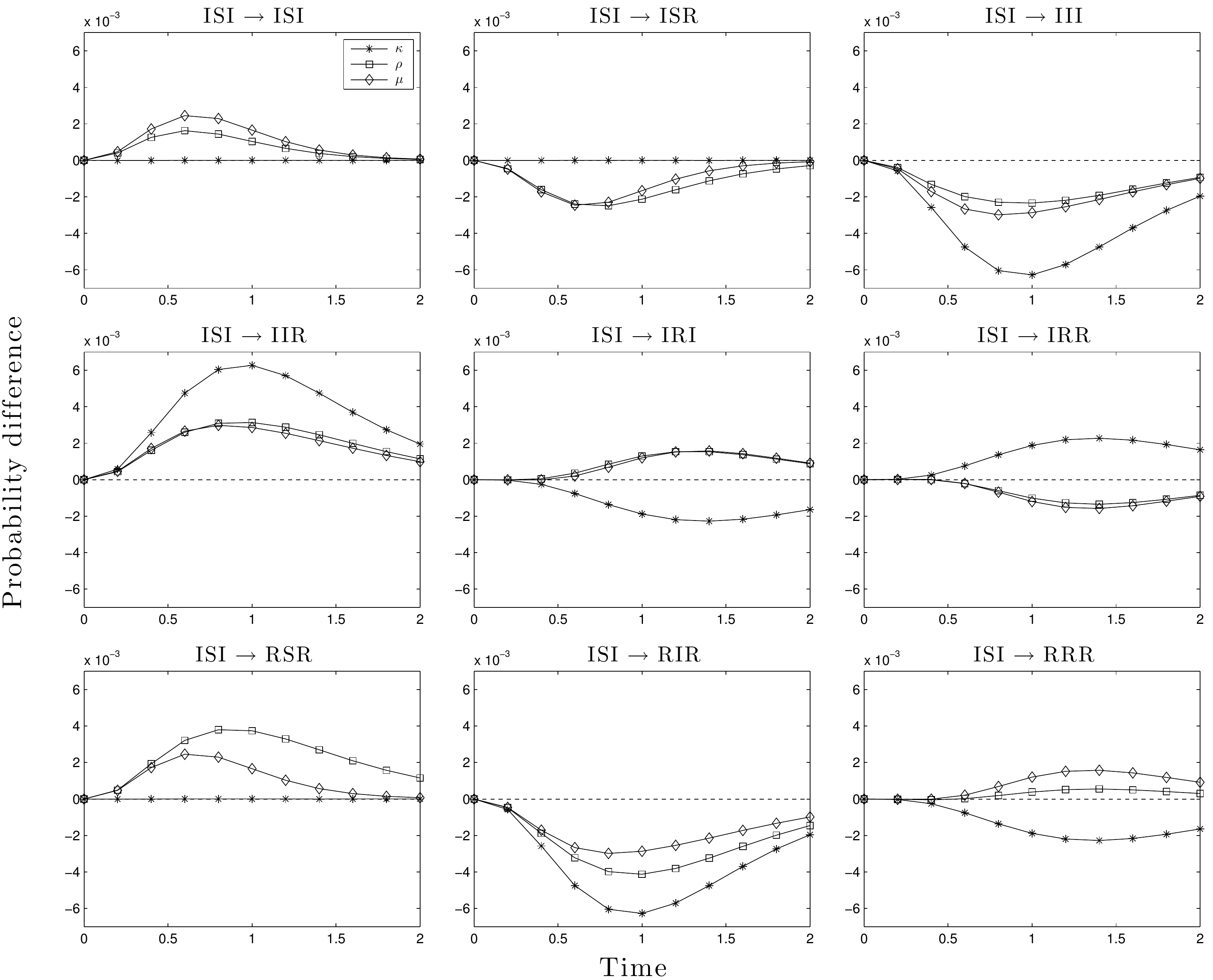}
\caption{Error $e_c^{\bx^0}(\bx;t)=\pb^{\bx^0}(\bx;t) - \pb_c^{\bx^0}(\bx;t)$ between the exact and approximate probabilities of a closed triangle being in state $\bx$ as a function of time, when starting from state $\bx^0$ at time $t=0$, for various choices of $\bx^0$ and $\bx$. The model is Markovian with infectivity $\tau=1$.}\label{fig:ErrorVStimeISI}
\end{figure}


\newpage

\begin{figure}[H]
\centering
\includegraphics[width=\textwidth]{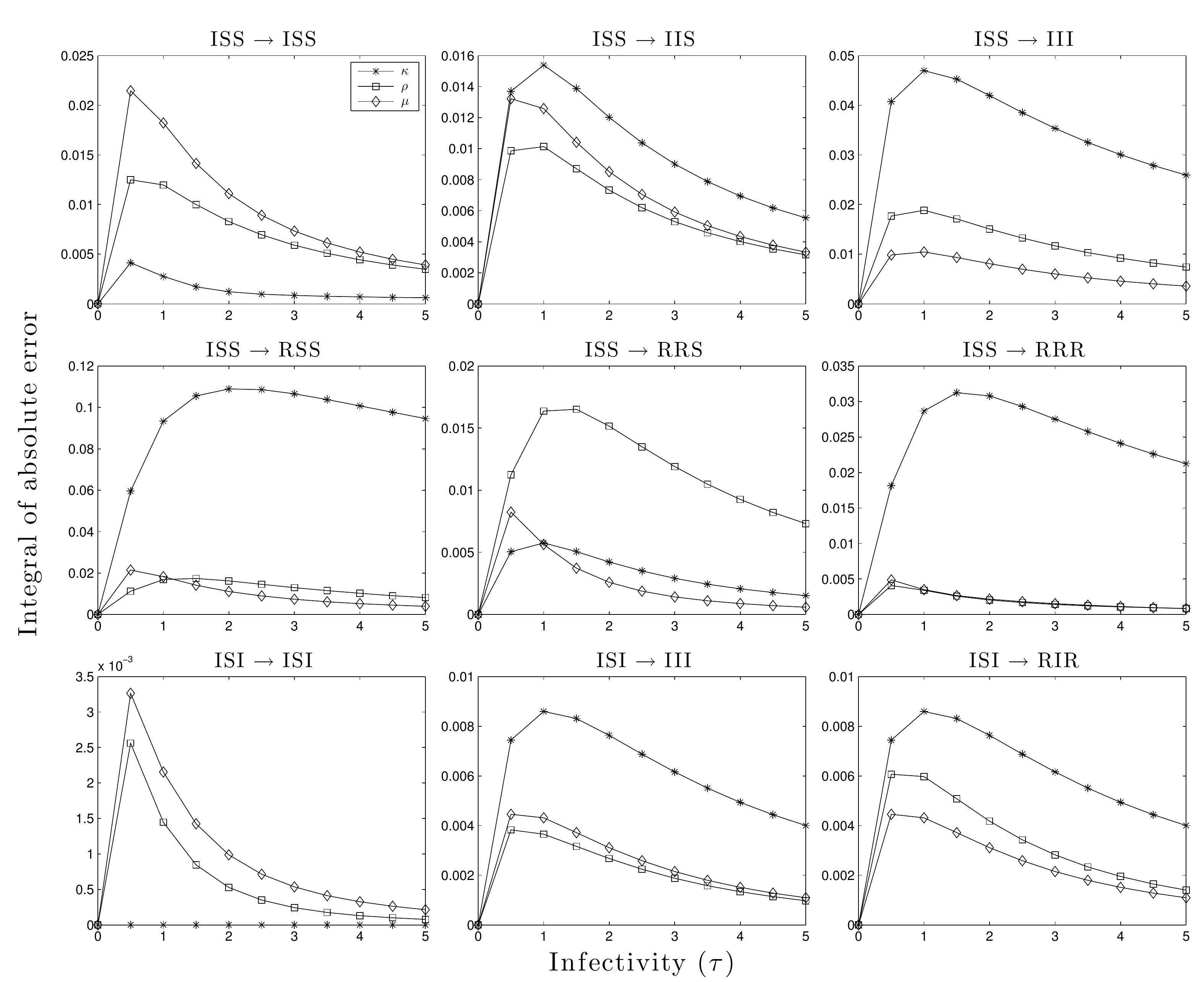}
\caption{Time integral of the modulus of the difference $e_c^{\bx^0}(\bx;t)=\pb^{\bx^0}(\bx;t) - \pb_c^{\bx^0}(\bx;t)$ between the exact and approximate probabilities of a closed triangle being in state $\bx$, when starting from state $\bx^0$ at time $t=0$, as a function of $\tau$. The Markovian SIR model is assumed
.}\label{fig:iad_VS_tau_bestc_L1_diffy} 
\end{figure}

\begin{figure}[H] 
\centering
\includegraphics[width=\textwidth]{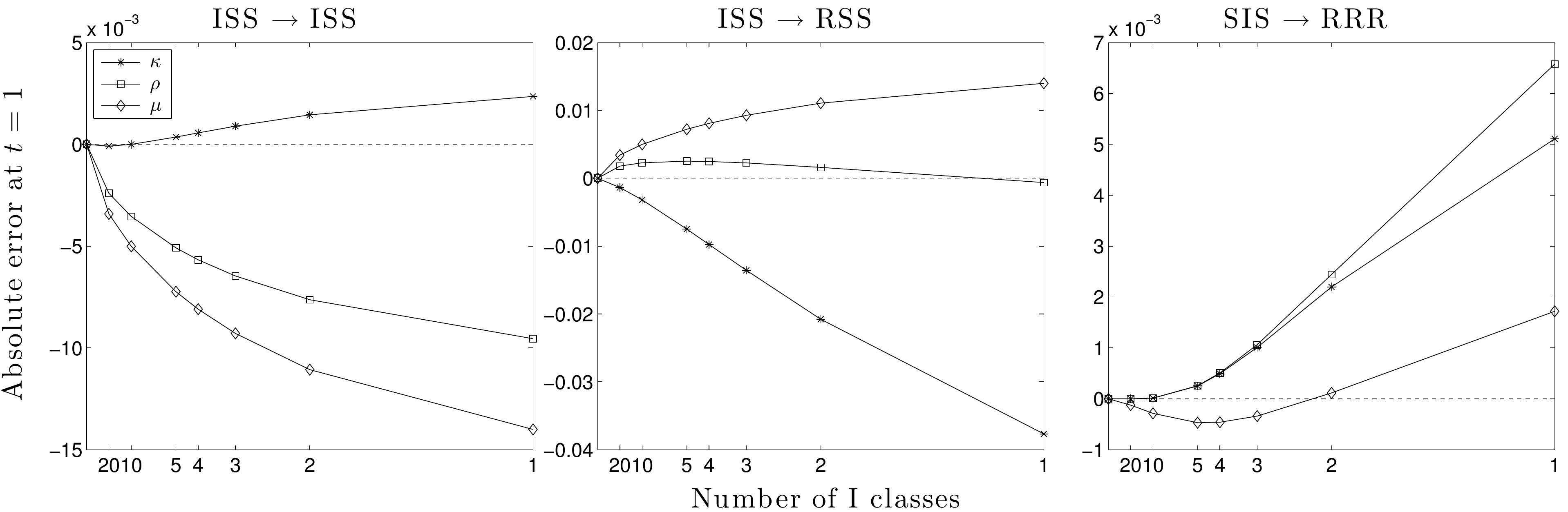}
\caption{Errors $e_c^{\bx^0}(\bx;t)=\pb^{\bx^0}(\bx;t) - \pb_c^{\bx^0}(\bx;t)$
($c = \kappa, \rho$ and $\mu$) between the exact and approximate probabilities
of an open triplet being in state $\bx$ at time $t=1$, when starting from state
$\bx^0$ at time $t=0$, for some selected choices of $\bx^0$ and $\bx$, as a
function of the number of infectious classes ($x$-axis linearly increasing with the variance),
for infectivity $\tau=1$.
}\label{fig:ClosedHeterogeneity}
\end{figure}

\newpage


\newpage

\begin{figure}[H]
\centering
\includegraphics[width=\textwidth]{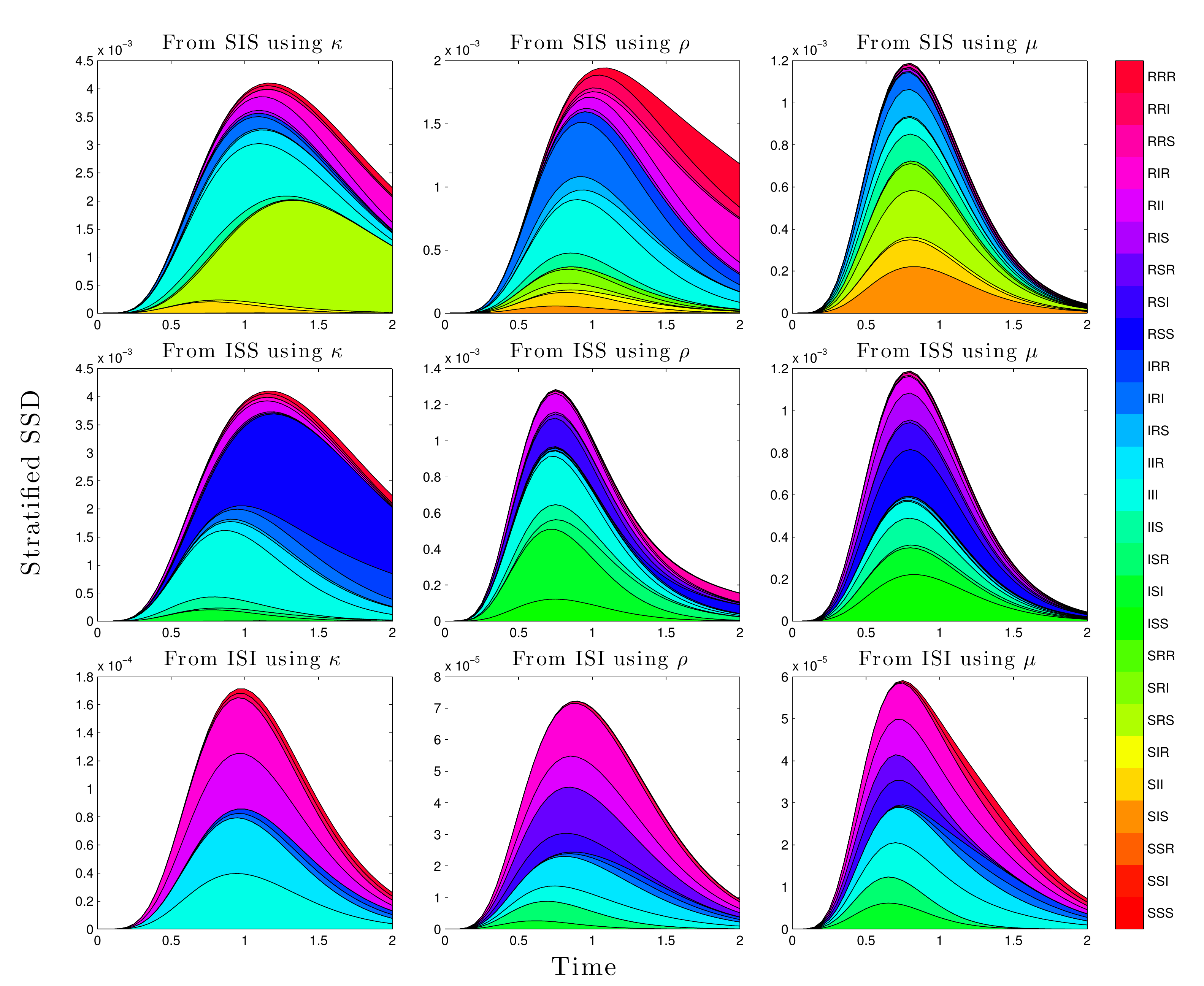}
\caption{Stratified contribution of each state $\bx$ to the overall discrepancy between the exact distribution over system states and each of the approximations (top row: Kirkwood's approximation; middle row: 1-step ME; third row: full ME), for different starting points $\bx^0$ (one per column). The model is Markovian with infectivity $\tau=1$.}\label{fig:ClosedStackedSSD_9plots_diffy} 
\end{figure}

\newpage

\begin{table}[H]
\centering
\includegraphics[height = 12.5cm]{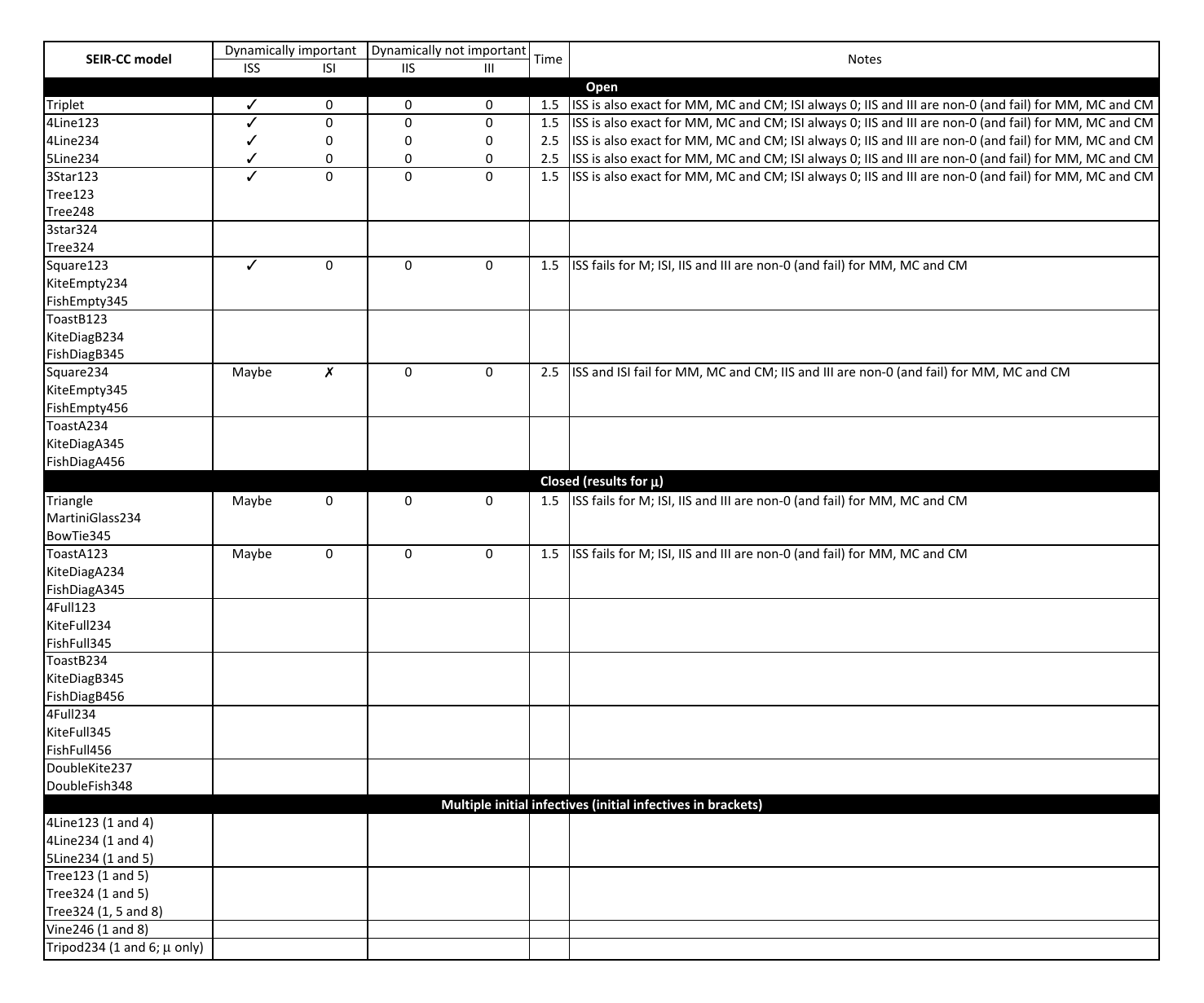}
\caption{Exactedness of moment closures at the level of triplets for the SEIR model with a constant duration of the latent and the infectious periods (CC). Comments for the models where either the latent of the infectious period, or both (MC, CM or MM, respectively) have exponentially distributed duration are also reported when useful. Time of test is $t=0.5$ when not stated. Only the interesting results are reported. Symbols and table structure are as per Table 1 in the main text.}
\label{tab:SEIR}
\end{table}

\newpage

\begin{figure}[H]
\centering
\includegraphics[width = \textwidth]{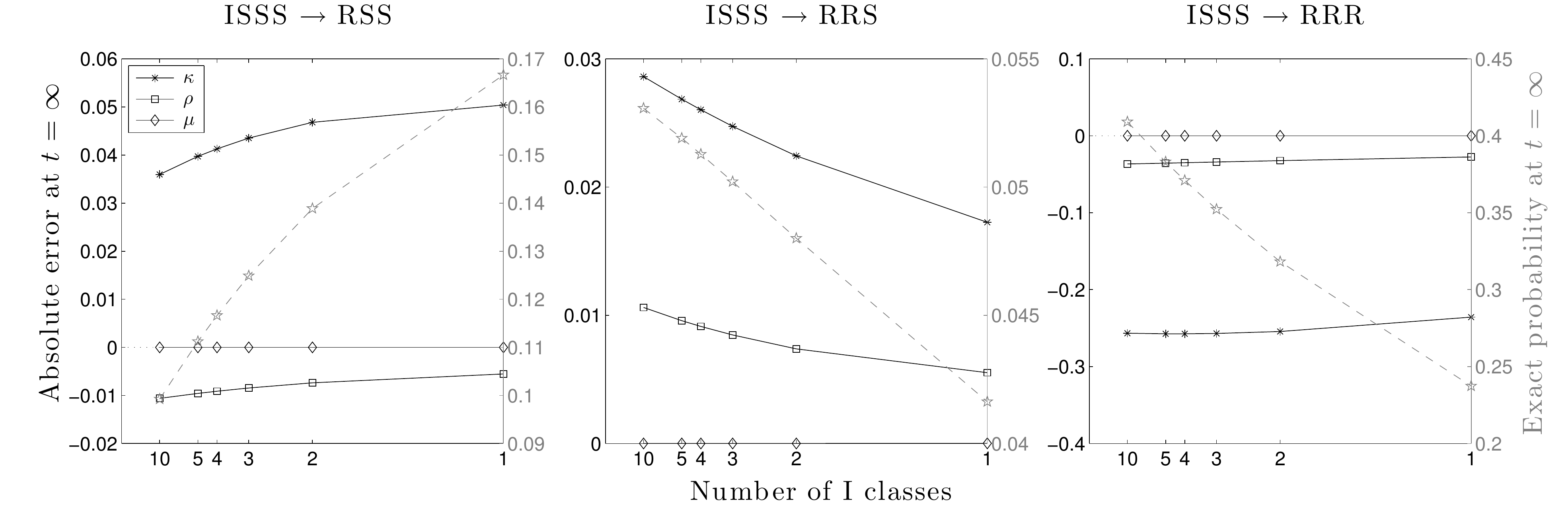}
\caption{Investigation of the error in the three moment closures (Kirkwood, $\kappa$; 1-step ME, $\rho$; and ME, $\mu$) for the SIR model on the MartiniGlass234 in its absorbing states for $t\to\infty$. The left axes (grey dashed lines with 5-point star markers) shows the probability that the system ultimately ends in the state of interest. Note the exactedness of ME as opposed to the other closures (see main text).}
\label{fig:FinalState_MartiniGlass_SIR}
\end{figure}

\begin{figure}[H]
\centering
\includegraphics[width = \textwidth]{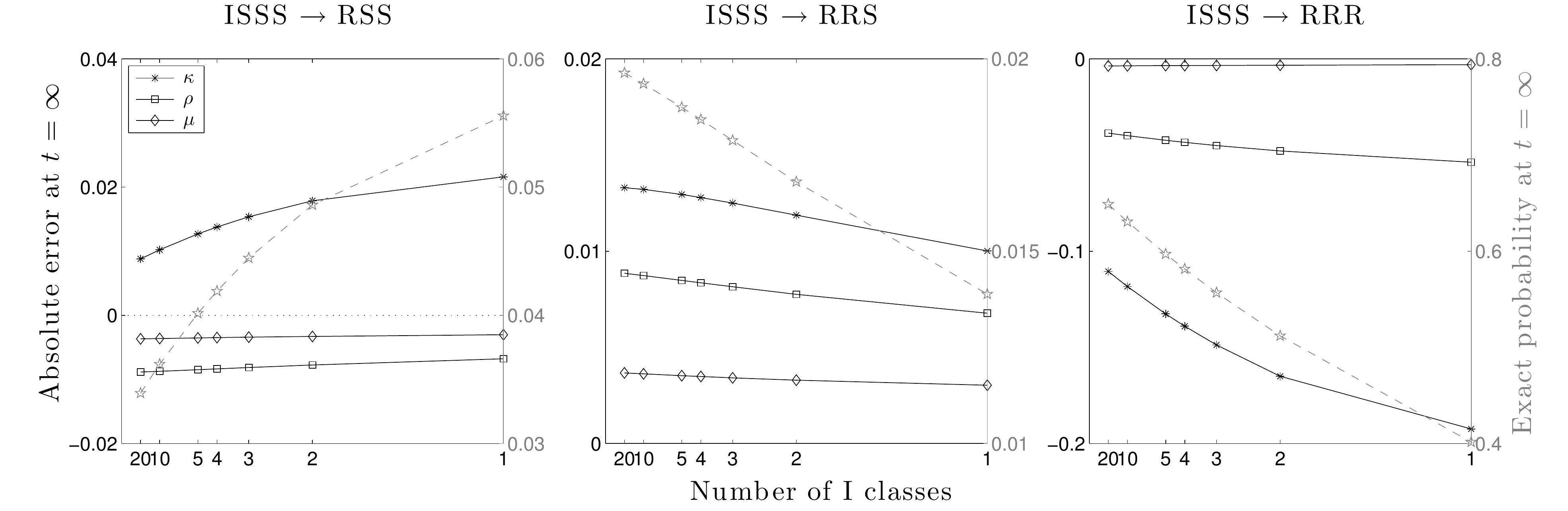}
\caption{Investigation of the error in the three moment closures (Kirkwood, $\kappa$; 1-step ME, $\rho$; and ME, $\mu$) for the SIR model on the ToastB234 in its absorbing states for $t\to\infty$. The left axes (grey dashed lines with 5-point star markers) shows the probability that the system ultimately ends in the state of interest. Note how all closure fails (see main text).}
\label{fig:FinalState_ToastB_SIR}
\end{figure}

\newpage

\section*{Supplementary Code}

\subsection*{Main Function}

{\scriptsize
\begin{verbatim}


function [Time,Y]= pair_based_me(T,g,I0,TSPAN,tol,niter)
% Modified version of code from Sharkey (2011) to use Maximum Entropy
% rather than Kirkwood closure

T=T';
N=length(T(:,1));
I_vec=zeros(N,1);I_vec(I0)=1;
S_vec=ones(N,1);S_vec(I0)=0;
I_mat=spdiags(I_vec,0,N,N);
S_mat=spdiags(S_vec,0,N,N);
G=spones(T);
G_A=G.*(1-G');
H=G+G_A';
Q=min(H^2,1);Q=Q-diag(diag(Q));
H_c=Q.*H;
H_o=Q-H_c;
F_AB=H;
F_AA=tril(H);
IS=I_mat*F_AB*S_mat;
SS=S_mat*F_AA*S_mat;
II=I_mat*F_AA*I_mat;
W_AB=find(reshape(F_AB,N^2,1));
W_AA=find(reshape(F_AA,N^2,1));
d_AB=length(W_AB);
d_AA=length(W_AA);
Y0=[S_vec;I_vec;IS(W_AB);SS(W_AA);II(W_AA)];
options=odeset('abstol',tol(1),'reltol',tol(2));
[Time,Y]=ode23(@model_function,TSPAN,Y0,options,G,T,H_c,H_o,g,N,W_AA,W_AB,d_AA,d_AB,niter);
end
\end{verbatim}
}

\clearpage

\subsection*{ODE function}

{\scriptsize
\begin{verbatim}


function dY = model_function(~,Y0,G,T,H_c,H_o,g,N,W_AA,W_AB,d_AA,d_AB,niter)

IS=spalloc(N,N,d_AB);
SS=spalloc(N,N,d_AB);
II=spalloc(N,N,d_AB);
S=Y0(1:N);
I=Y0(N+1:2*N);
IS(W_AB)=Y0(2*N+1:2*N+d_AB);
SS(W_AA)=Y0(2*N+d_AB+1:2*N+d_AB+d_AA);
II(W_AA)=Y0(2*N+d_AB+d_AA+1:2*N+d_AB+2*d_AA);
SS=SS+SS';
II=II+II';
inv_S=spdiags(spfun(@inve,S),0,N,N);
inv_I=spdiags(spfun(@inve,I),0,N,N);

R=T.*IS;

% Other code is unchanged from Sharkey (2011); the below runs niter iterations
% of the interative method for calculating the maximum entropy distribution

IrSrS = R'*H_o.*(inv_S*SS);
IrSlI = IS*inv_S.*(H_o*R);
[ii,jj] = find(H_c);
for t=1:length(ii)
    i=ii(t); j=jj(t);
    kk = full(intersect(find(G(i,:)),find(G(j,:))));
    for k=kk
        aP12 = [full(SS(i,j)), full(IS(j,i)); full(IS(i,j)), full(II(i,j))];
        aP23 = [full(SS(j,k)), full(IS(k,j)); full(IS(j,k)), full(II(j,k))];
        aP13 = [full(SS(i,k)), full(IS(k,i)); full(IS(i,k)), full(II(i,k))];
        Tri = MaximumEntropy(aP12, aP23, aP13, 2, niter);
        IrSrS(i,j) = IrSrS(i,j) + Tri(1,1,2);
        IrSlI(i,j) = IrSlI(i,j) + Tri(2,1,2);
    end
end

SrSlI=IrSrS';
IrSrI=IrSlI';

dT=sum(R)';
dS=-dT;
dI=dT-g*I;
dSS=-IrSrS-SrSlI;
dIS=IrSrS-IrSlI-R-g*IS;
dII=IrSlI+IrSrI+R+R'-2*g*II;

dY=[dS;dI;dIS(W_AB);dSS(W_AA);dII(W_AA)];

end
\end{verbatim}
}

\clearpage

\subsection*{Helper functions}

{\scriptsize
\begin{verbatim}

function Ptriplet = MaximumEntropy(P12,P23,P13,ns,niter)

P = zeros(ns,ns,ns,niter);
Ptemp_old = zeros(ns,ns,ns);
Ptemp_new = zeros(ns,ns,ns);
P(:,:,:,1) = 1/ns^3;

for ni = 2:niter
    Ptemp_old = P(:,:,:,ni-1);
    for g1 = 1:ns
        for g2 = 1:ns
            for g3 = 1:ns
                den = sum(Ptemp_old(g1,g2,:));
                if den == 0
                    Ptemp_new(g1,g2,g3) = 0;
                else
                    Ptemp_new(g1,g2,g3) = P12(g1,g2) * Ptemp_old(g1,g2,g3) / den;
                end
            end
        end
    end
    testtemp = sum(sum(sum(Ptemp_new)));
    Ptemp_old = Ptemp_new;
    for g1 = 1:ns
        for g2 = 1:ns
            for g3 = 1:ns
                den = sum(Ptemp_old(:,g2,g3));
                if den == 0
                    Ptemp_new(g1,g2,g3) = 0;
                else
                    Ptemp_new(g1,g2,g3) = P23(g2,g3) * Ptemp_old(g1,g2,g3) / den;
                end
            end
        end
    end
    testtemp = sum(sum(sum(Ptemp_new)));
    Ptemp_old = Ptemp_new;
    for g1 = 1:ns
        for g2 = 1:ns
            for g3 = 1:ns
                den = sum(Ptemp_old(g1,:,g3));
                if den == 0
                    Ptemp_new(g1,g2,g3) = 0;
                else
                    Ptemp_new(g1,g2,g3) = P13(g1,g3) * Ptemp_old(g1,g2,g3) / den;
                end
            end
        end
    end
    testtemp = sum(sum(sum(Ptemp_new)));
    P(:,:,:,ni) = Ptemp_new;
end
Ptriplet = P(:,:,:,niter);
end
\end{verbatim}

\begin{verbatim}


function M_out=inve(M_in)

M_out=M_in.^(-1);

end
\end{verbatim}
}

\end{document}